\documentclass[5p]{elsarticle}
\usepackage{lineno,hyperref}
\modulolinenumbers[1]  

\usepackage{tikz}
\usepackage{pgfplots}
\usetikzlibrary{plotmarks,arrows, positioning, calc, patterns, shapes, trees,
decorations.pathreplacing,plotmarks}
\newlength\figureheight
\newlength\figurewidth

\usepackage{enumitem}

\usepackage{algorithm}
\usepackage{algorithmic}

\usepackage{array}	

\usepackage{mdwmath}
\usepackage{mdwtab}
\usepackage{amsthm}
\newtheorem{theorem}{Theorem}
\newtheorem*{theorem*}{Theorem}

\newtheorem{proposition}[theorem]{Proposition}

\usepackage{graphicx}
\usepackage{subfig}
\usepackage{stfloats}
\usepackage{fixltx2e} 
\usepackage{url}

\usepackage{graphics} 
\usepackage{epsfig} 
\usepackage{mathptmx} 
\usepackage{times} 
\usepackage{setspace}
\usepackage{amsmath}
\usepackage{amssymb}
\usepackage{pgfplots, pgfplotstable}
\usetikzlibrary{patterns}
\usepackage{url}
\usepackage{longtable}

\makeatletter
\def\ps@pprintTitle{%
 \let\@oddhead\@empty
 \let\@evenhead\@empty
 \def\@oddfoot{}%
 \let\@evenfoot\@oddfoot}
\makeatother

\begin{document}

\begin{frontmatter}

\title{ An efficient null space inexact Newton method for  hydraulic simulation of water distribution networks}

\author[EAbAddr]{Edo~Abraham\corref{cor1}\fnref{fn1}}
 \ead{edo.abraham04@imperial.ac.uk}
 \author[EAbAddr]{Ivan~Stoianov\fnref{fn2}}
 \ead{ivan.stoianov@imperial.ac.uk}
\cortext[cor1]{Corresponding author} 
\fntext[fn1,fn2]{This work was supported by the  NEC--Imperial Smart Water Systems project. \\
Post-print extension of [17],  September 2015. Includes extended exposition, additional case studies and new simulations and analysis.} 
\address[EAbAddr]{Dept. of Civil and Environmental Engineering, Imperial College London, London, UK.}

\begin{abstract}
Null space Newton algorithms are efficient in solving the nonlinear equations arising in hydraulic analysis of water distribution networks. In this article, we propose and evaluate an inexact Newton method that relies on partial updates of the network pipes' frictional headloss computations to solve the linear systems more efficiently and with numerical reliability. The update set parameters are studied to propose appropriate values. Different null space basis generation schemes are analysed to choose  methods for sparse and well-conditioned null space bases resulting in a smaller update set. The Newton steps are computed in the null space by solving sparse, symmetric positive definite systems with sparse Cholesky factorizations. By using the constant structure of the null space system matrices, a single symbolic factorization in the Cholesky decomposition is used multiple times, reducing the computational cost of linear solves. The algorithms and analyses are validated using medium to large-scale water network models.

\end{abstract}

\begin{keyword}
Null space algorithm, inexact Newton method, partial loop flow updates, hydraulic analysis, sparse LU
\end{keyword}

\end{frontmatter}


\section{Introduction}\label{sec:Introduction}
Advances in sensor, automatic control and information technologies  have  enabled the solution of increasingly challenging operational problems for smarter water distribution networks (WDNs). Reliable and efficient  tools for  modelling,  estimation, optimal control, incident/fault detection, and design optimization for large-scale hydraulic systems are vital to solving, in near real time and for progressively larger networks, challenges arising from growing water demand, ageing water infrastructure and more stringent environmental standards.
An extensive overview of operational, technical and economical  challenges facing water utilities, and a collection of current research problems can be found in~\cite{sensus2012,ccwi2013}, respectively, and the references therein. Hydraulic analysis  is essential in all these; a set of nonlinear equations governing pipe flows and nodal pressures across the network are solved to simulate the water distribution system behaviour. 
For example, optimal network rehabilitation/design problems include the optimal choice of pipes and control valves,  and their number and location under economic constraints. Previous work in literature has coupled conventional hydraulic simulation tools like EPANET~\cite{rossman2000epanet} with heuristic optimization schemes (eg. evolutionary algorithms) to solve these network design problems~\cite{Maier2014EA,savic1997genetic,nicolini2009optimal}. The same nonlinear hydraulic equations are also employed in mathematical optimization approaches for optimal network pressure control problems~\cite{Wright2014,eck2013fast}. Therefore, savings in computational time of hydraulic analysis are important to many an optimization problem for WDNs. 

 This article is concerned with  demand-driven hydraulic analysis~\cite{guidolin2013using}, which poses the flow continuity and energy conservation laws for a pipe network as a set of nonlinear equations of the flows and unknown pressure heads for given nodal demands.  The Newton method for solving nonlinear equations was exploited by~\citep{epp1970efficient} to pose  an iterative hydraulic solver, and some years later coupled with a preconditioned conjugate gradient linear solver and called Global Gradient Algorithm (GGA) in~\cite{todini1988section}.   As the size of networks modelled by water utilities become larger, various  approaches  have been proposed in recent literature  to improve computational efficiency of the GGA method. 
Some work has considered the reduction of the mathematical problem through a smaller topological representation of the original water network model; it has been standard practice for water utilities to skeletonize networks so each node abstracts an entire area or multiple  points of consumption~\cite{bhave1988calibrating,thamesWaterAtkins2013}. For example, a new method for lumping of serial nodal demands along a pipe while maintaining sufficient accuracy in the energy balance is proposed in~\cite{giustolisi2011computationally}.  In applications where multiple simultaneous simulations of networks are required, parallelizing at the  level of the analysis software using clusters of computers, multiple core CPUs, or GPUs has been shown to give promising speedups~\cite{mair2014performance,crous2012potential}. On a finer grain, parallelization of headloss computations in individual hydraulic simulation steps are employed in~\cite{guidolin2013using} to reduce computational time. Although the most significant percentage of computational time is used by the linear solver at each Newton iteration, the sequential data access by the linear algebra operations makes it less suitable for parallelism~\cite{guidolin2013using}. As the bottleneck of the Newton method for solving hydraulic equations, efficiently solving the linear systems is paramount and so is the subject of this article.

The Newton method for hydraulic analysis has a Jacobian with a saddle point structure~\cite{abrahamnull2014,benzi2005numerical}. In the numerical optimization literature, null space algorithms for saddle point problems have been used extensively, often called {\it reduced Hessian} methods~\cite{benzi2005numerical}. Null space algorithms, as opposed to the range space approach of GGA~\cite{todini2012unified}, have also been applied for hydraulic analysis of water and gas pipe networks~\cite{nielsen1989methods,rahal1995cotree,elhay2014reformulated,abrahamnull2014}. For a WDN with $n_p$ number of pipes (or links) and $n_n$ unknown-head nodes, the number $n_l=n_p-n_n$, which is the number of co-tree flows~\cite{elhay2014reformulated}, is often much smaller than $n_n$.  At each iteration, whereas the GGA method solves a linear problem of size $n_n$, a null space method solves an often much smaller problem of size $n_l$ but with the same symmetric positive definiteness properties.  Therefore,  significant computational savings can be made for sparse network models.
Moreover, GGA becomes singular when one or more of the head losses vanish. Unlike the GGA approach, null space algorithms do not involve inversion of headloss values. As  such, they will not require processes to deal with zero flows so long as there are no loops with all zero flows~\cite{abrahamnull2014,benzi2005numerical,elhay2014reformulated}.

%

In this article, which is an extended exposition of the post-print from~\cite{abrahamnull2014}, we investigate further the null space Newton algorithms for hydraulic analysis proposed in~\cite{abrahamnull2014}.
 By using sparse null space basis, we show that  a significant fraction of the network pipes need not be involved in the flow updates of the null space Newton method. In addition to these, we take advantage of the loop flows that converge fast to propose a partial update scheme that reduces the number of computations in calculating head losses and matrix-matrix multiplications. By formulating the partial updates as an inexact Newton method, the method guarantees nice convergence properties. We also study the Newton tolerance and partial update set parameters to suggest appropriate parameter values.    
Since the flow update equations of the null space algorithm do not depend on pressure evaluations, a heuristic for reducing the number of pressure head computations is utilised for further computational savings. We demonstrate through case studies that, for sparse network models,  the proposed null space solvers can reduce CPU time by up to a factor of 4 compared to GGA. 

We first present a step-by-step derivation of the null space algorithm from the hydraulic equations, and then  discuss various computational tools for generating sparse null bases and sparse factorizations. State of the art solvers from the SuiteSparse library~\cite{davis2004umfpack,davis2011spqr} are used. 
In our implementation, the values that stay constant over different steady state simulations are computed only once. In addition to the Hazen-William pipe resistance computations, the matrix whose columns span the null space of the network topology and the Cholesky factors for the head equations are two more examples.  Within each hydraulic simulation, values that do not change at each iteration are also solved for only once. We show  the iterative Darcy-Weisbach and rational exponent   Hazen-William head losses do not need to be recomputed for pipes not involved in the loop equations of the null space algorithm. 
 
The remainder of this article is organised as follows. In the next section, we will discuss the hydraulic analysis problem and traditional solution methods. 
Section~\ref{sec:nullspacealgs} examines the structure of the Newton linear systems and then discusses relevant null space algorithms. Sparse null basis computation tools are also discussed and implemented. In Section~\ref{sec:parUpdate}, novel methods for reducing the computational cost of the null space algorithm are presented. The use of partial update sets and a related  new  method  for reducing head computations are described. Mathematical proofs are presented to show the Newton method stays convergent with the introduced modifications. Finally, a  numerical study with further results is presented using a number of operational and modified network examples detailed in Section~\ref{sec:caseStudy}, followed by our conclusions in Section~\ref{sec:conclusion}. 

{\bf Notation: } For a vector $v\in \mathbb{R}^n$, we define the usual p-norms as  $\lVert v \rVert_p:=(\sum_{i=1}^{n}{|v_i|^p})^{1/p}, p=1,2$ and ${\|v\|_p= \max\limits_{i} |v_i|}$ if $p=\infty .$  
For a matrix $A$, $||A||_p=\max\limits_{||x||=1}\frac{||Ax||_p}{||x||_p}$, where $||Ax||_p$, $||x||_p$ are the corresponding vector norms.
$A^T$ denotes the transpose of the matrix $A.$ For an invertible matrix $X$, we denote its condition number  by $\kappa(X)_p:=\|X\|_p\|X^{-1}\|_p.$ The (right) null space of a matrix $A$ is also denoted by $\text{ker}(A).$

%
%
%
%
%
%
%
%
%
%
%
%
%
%
%
%
%
%
%
%
%
%

%
%
%



\section{Flow continuity and energy conservation equations: solution via the Newton method}


In this article, we deal with demand-driven hydraulic analysis, where the demand is assumed known. In contrast, pressure-driven demand and leakage simulations represent demands as nonlinear functions of pressure~\cite{giustolisi2008pressure} to be solved for. For a  network with $n_p$  links connecting ${n_n (<n_p)}$ unknown head junctions, and $n_0$ known head junctions, we define the vector of unknown flows and pressure heads as ${q=[q_1,\ldots,q_{n_p}]^T}$ and ${h=[h_1,\ldots,h_{n_n}]^T},$ respectively. Let pipe $p_j$ have flow $q_j$ going from node $i$ to node $k$, and with pressure heads $h_i$ and $h_k$ at nodes $i$ and $k$, respectively. The frictional headloss (or flow resistance)  across the pipe can then be represented as:
\begin{equation}\label{eq:pipeloss}
  h_i - h_k = r_j|q_j|^{n-1}q_j,
\end{equation}
where $r_j$, the resistance coefficient of the pipe, can be modelled as either independent of the flow or implicitly dependent on flow $q_j$ and given as $r_j=\alpha L_j/(C_j^nD_j^m)$. The variables $L_j,$ $D_j$ and $C_j$ denote the length, diameter and roughness  coefficient of pipe $j$, respectively.  The triplet $\alpha,$ $n$ and $m$ depend on the energy loss model used; Hazen-Williams (HW: $r_j=10.670 L_j/(C_j^{1.852} D_j^{4.871})$) and Darcy-Weisbach (DW) are two commonly used frictional head loss formulae~\cite{elhay2011dealing}. In DW models, the dependence of the  resistance coefficient on flow is implicit; see the formulae in~\cite[(1)--(2)]{simpson2010jacobian}.
With head loss equations defined for each pipe and the fixed heads and demands for each node taken into account, the steady-state fluid flows in a water network must satisfy the two hydraulic principles:
\begin{align}
A_{12}^T q  -d=0,\label{eq:flowcont}\\
A_{11}(q)q + A_{12}h+A_{10} h_0=0,\label{eq:energyconsrv}
\end{align}
where the variables $h_0\in\mathbb{R}^{n_0}$ and $d\in\mathbb{R}^{n_n}$ represent the known heads (eg.\ at a reservoir or tank) and demands at nodes, respectively. While~\eqref{eq:flowcont} guarantees the conservation of flow at each junction node,~\eqref{eq:energyconsrv} accounts for the frictional head loss across all links. Here, the matrices $A_{12}^T\in\mathbb{R}^{n_n\times n_p}$ and $A_{10}^T\in\mathbb{R}^{n_n\times n_0}$ are the node-to-edge incidence matrices for the  $n_n$ unknown head nodes and $n_0$ fixed head nodes, respectively. For example, each link is associated with an $n_n\times 1$ row vector in $A_{12}$: $A_{12}(j,i)=1 (\text{ or }-1)$ if link $j$ enters ( or leaves) node $i$ and $A_{12}(j,i)=0$ otherwise.
The square  matrix $A_{11}\in\mathbb{R}^{n_p\times n_p}$ is a diagonal matrix with the elements 
\begin{equation}\label{eq:loss_diag}
A_{11}(j,j)=r_j|q_j|^{n_j-1},j=1,\ldots,n_p,
\end{equation}
representing part of the loss formula in~\eqref{eq:pipeloss}. 
The set of nonlinear equations~\eqref{eq:flowcont} and~\eqref{eq:energyconsrv} can be represented by the matrix equation~\cite[Eq.\ (1)]{todini1988section}: 
\begin{equation}
\label{eq:hydro_nl_eqn}
f(q,h) := \begin{pmatrix}
A_{11}(q) & A_{12}\\
A_{12}^T & 0
\end{pmatrix} 
\begin{pmatrix}
q\\
h
\end{pmatrix} +
\begin{pmatrix}
A_{10} h_0\\
-d
\end{pmatrix} =0.
\end{equation}

Most  non-linear equations  and unconstrained optimization problems are solved using Newton's method~\cite{nocedal2006numerical,dennis1996numerical}. The same Newton method has been applied to solve  hydraulic analysis problems, as early as in~\cite{epp1970efficient}, and has been extensively used for the same purpose since then. 
By considering  the Jacobian of $f(q,h)$ with respect to the unknown $x := [q^T \;\; h^T]^T$, and using the head loss model in~\eqref{eq:loss_diag}, the Newton iteration for the solution of~\eqref{eq:hydro_nl_eqn} is~\cite{abrahamnull2014}:
\begin{equation}\label{eq:nr_lin_eq}
\begin{aligned}
\nabla f(x^k) (x^{k+1}-x^k) &=-  f(x^k)   \\
\begin{bmatrix}
N A_{11}(q^k) & A_{12}\\
A_{12}^T & 0
\end{bmatrix} 
\begin{bmatrix}
dq\\
dh
\end{bmatrix} &= -
 \begin{bmatrix}
A_{11}(q^k) & A_{12}\\
A_{12}^T & 0
\end{bmatrix} \begin{bmatrix}
q^k\\
h^k
\end{bmatrix} \\&+
\begin{bmatrix}
-A_{10} h_0\\
d
\end{bmatrix} 
\end{aligned}
\end{equation}
where $\begin{bmatrix}dq\\ dh \end{bmatrix}=\begin{bmatrix}q^{k+1}-q^k\\ h^{k+1} - h^k \end{bmatrix}$ and $N=\texttt{diag}(n_i), \; i=1,\ldots,n_p.$ 

In~\eqref{eq:hydro_nl_eqn}, popularly called the Global Algorithm formulation~\cite{todini2012unified,todini1988section}, the frictional headloss function is expressed as a function of the flows.  Using a nonlinear transformation of the headloss in pipe $p_j$,~\eqref{eq:pipeloss} can be reformulated to:
\begin{equation}\label{eq:pipelossheads}
 q_j =  r_j^{-1/n}|\Delta h_j|^{\frac{1-n}{n}}\Delta h_j, \quad  \Delta h_j= h_i - h_k,
\end{equation}
where pipe $p_j$ is topologically represented as going from node $i$ to node $k.$ A matrix form of~\eqref{eq:pipelossheads} is
\begin{equation}\label{eq:pipelossheadsMat}
 q =  \hat{A}_{11} (A_{12} h+A_{10} h_0), 
\end{equation}
where $\hat{A}_{11}(j,j)=r_j^{-1/n}|\Delta h_j|^{\frac{1-n}{n}},j=1,\ldots,n_p.$ Substituting~\eqref{eq:pipelossheadsMat} in the continuity equations~\eqref{eq:flowcont}, what is called the `nodal head representation'~\cite{todini2012unified} of the hydraulic equations is projected to the size $n_n$ nonlinear equations:
\begin{equation}\label{eq:hydro_nl_NHeqn}
 A_{12}^T\hat{A}_{11} (A_{12} h+A_{10} h_0)-d=0, 
\end{equation}

Although the smaller number of nonlinear equations~\eqref{eq:hydro_nl_NHeqn}, in the unknowns $h,$ can be solved using Newton's method, it has been shown via case studies~\cite{epp1970efficient,todini2012unified} that the Newton iterations on the new nonlinearities (i.e.\ formulating the flows in terms of  energy heads only) take many more iterations than when Newton's  method is applied to~\eqref{eq:hydro_nl_eqn}. In addition to more Newton iterations, this nodal formulation does not result in linear systems with saddle point structure, which would allow for the use of faster and better conditioned null space methods~\cite{abrahamnull2014,elhay2014reformulated}. 
   Here, we first present the Newton method for solving~\eqref{eq:hydro_nl_eqn} by stating its convergence properties. The following proposition is used to guarantee convergence properties of a partial update null space method by posing it as inexact Newton method.

{\begin{proposition}(Convergence of Newton method for hydraulic analysis)

Let $x^* := [q^*\;\; h^*]^T\in D$, with open convex set $D$, be a non-degenerate solution of~\eqref{eq:hydro_nl_eqn}, i.e.\ the Jacobian $\nabla f(x^*)$ is not singular, and let $\{x^k\}$ be the sequence of states generated by the Newton iteration~\eqref{eq:nr_lin_eq}. For $x^k\in D$ sufficiently near $x^*$, the Newton sequence exists (i.e.\ $\nabla f(x^i)$ is nonsingular for all $i>k$) and has local superlinear convergence.
\label{prop:ENM}
\end{proposition}}

\begin{proof}(\cite[Lemma 1]{abrahamnull2014})
 As we show in~\cite{abrahamnull2014},  $f(\cdot)$ is continuously differentiable in $\mathbb{R}^{n_p+n_n}$ for both Darcy-Weisbach and Hazen-Williams models. If we assume $x^*$ is non-degenerate, the proof is a standard result and is relegated to~\cite[Thm.\ 11.2]{nocedal2006numerical}.
\end{proof} 
The Newton method is often preferred as a nonlinear equation solver because of its asymptotic quadratic convergence. We show in~\cite{abrahamnull2014} that the Jacobians of $f(\cdot)$ in~\eqref{eq:hydro_nl_eqn} are Lipschitz either when the appropriate Darcy-Weisbach equation~\cite[(1--2)]{simpson2010jacobian} is used or when regularised Jacobians are used for zero flows and small laminar flows in Hazen-Williams models~\cite{gorev2012method}. With the use of such models to cope with zero flows, the Newton algorithm will have local quadratic convergence by~\cite[Thm.\ 11.2]{nocedal2006numerical}.

Almost all of the computational cost of the Newton method is incurred in the repeated solving of the linear system~\eqref{eq:nr_lin_eq} to find the Newton step. This linear system is, however, very sparse and has a special structure. Therefore, the rest of this article concerns  the structure of~\eqref{eq:nr_lin_eq} and proposes novel and efficient solvers based on linear transformations of this matrix  and tailored to the peculiarities of the hydraulic nonlinearities concerned.

\section{Null-space algorithms for hydraulic analysis}\label{sec:nullspacealgs}
\subsection{Problem formulation and algorithm derivations }
An interesting property of the Newton iteration linear equations~\eqref{eq:nr_lin_eq} is that they have what is called a {\itshape saddle point structure}~\cite{benzi2005numerical}; if the $2\times 2$ block structure is considered, the $A_{11}$ block is symmetric positive definite or semidefinite, $A_{21}=A_{12}^T\in \mathbb{R}^{n_n\times n_p},\; n_p\geq n_n,$ and $A_{22}=0$. The same class of problems arise in many PDE constrained optimization problems with various boundary conditions~\cite{pearson2012regularization}. Due to the indefiniteness  and often poor conditioning of the matrix, saddle point systems are challenging to solve efficiently and accurately. 
When the assumption that $A_{11}$ is invertible is valid, considering the block partitions of~\eqref{eq:nr_lin_eq} and applying block substitutions (a  Schur complement reduction~\cite[Sec.\ 5]{benzi2005numerical}), we  derive  an equivalent linear system with a smaller number of primary unknowns:
\begin{equation}
\begin{aligned}
A_{12}^T(NA_{11}^k)^{-1} A_{12}  h^{k+1} & =  -A_{12}^TN^{-1}(q^k + (A_{11}^k)^{-1}A_{10}h_0) -\\&(d-A_{12}^Tq^k)\label{eq:nr_lin_eq_1a}
\end{aligned}
\end{equation}
\begin{equation}
 q^{k+1} = (I-N^{-1})q^{k}-(NA_{11}^k)^{-1} ( A_{12} h^{k+1}+A_{10}h_0)\label{eq:nr_lin_eq_1b}.
\end{equation}
It is fortuitous here that, for invertible $A_{11}^k$, this  Schur complement reduction involves only simple element-wise inversions of the diagonal matrices $A_{11}^k$ and $N$ and the linear system stays sparse; for a general saddle point system, the Schur inversion can cause excessive fill-in even when the $A_{11}$ and $A_{12}$ blocks are sparse. 
Therefore, given an initial guess $(q^k, h^k)$, solving~\eqref{eq:nr_lin_eq} can be accomplished by first solving~\eqref{eq:nr_lin_eq_1a} for the pressure heads and  the flows  $q^{k+1}$ are then computed  by substituting for $h^{k+1}$ in~\eqref{eq:nr_lin_eq_1b}. In~\cite{todini1988section}, this reformulation of~\eqref{eq:nr_lin_eq} is called  `the nodal gradient formulation' of  GGA; this, simply because a linear system of equations is now solved only for the node heads in~\eqref{eq:nr_lin_eq_1a}.  Since the GGA method uses the Schur complement reduction of the larger saddle point matrix in~\eqref{eq:nr_lin_eq},   we also call GGA a Schur method/algorithm from here on, a name often used in the numerical analysis and optimization literature~\cite{benzi2005numerical}. 

One limitation of the Schur approach is the requirement that the $A_{11}$ block be nonsingular. When zero or very small flows occur due to either closed valves or zero demand in parts of the network,
 $A_{11}$ would be singular for head loss equations modelled by the Hazen-Williams formula. For the closed valve cases, by  expressing the headloss across  them   by a new variable and explicitly enforcing  a zero flow through them has been used to avoid singularities in~\cite{Boulos1993explicit}. However, in large operational networks  zero flows often arise due to zero demands at different time periods
and  as a result of  action by pressure and flow control devices~\cite{gorev2012method}.
In such cases, it is not known a priori which flows are zero; Figure~\ref{fig:Flowhist} shows a histogram of flows in the network BWKWnet  at peak hour, where some 5\% flows are shown to be zero and none of them  due to closed pipes;  see Section~\ref{sec:caseStudy} for details on case study networks. Therefore, other ways to deal with zero flows are needed.
 \begin{figure}%
 \centering
    \includegraphics[height=0.25\textwidth]{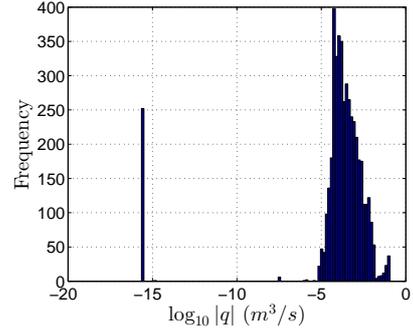} %
    \caption{A histogram showing the distribution of pipe flows for BWKnet network at 8:15 am. Here, the zero flows are set to machine precision ($\texttt{eps}, \approx 2e^{-16}$ on the CPU used) for plotting purposes here.}%
    \label{fig:Flowhist}%
\end{figure} 

Unlike in the Schur complement reduction, there is no requirement for $A_{11}$ to be nonsingular in a null space reformulation.
Assuming that $A_{12}$ has full column rank, which is shown to be true for WDN model in~\cite{elhay2014reformulated}, and ${\text{ker}(A_{11}) \cap \text{ker}(A_{12}^T)=\{0\}}$, a much smaller problem  can be solved at each iteration using null-space methods.
Let the columns of a non-zero matrix $Z\in \mathbb{R}^{n_p\times n_l}$, $n_l=n_p-n_n,$ span the null space of $A_{12}^T$, i.e.\ $A_{12}^TZ=0$, we can decompose $q^{k+1}$ in~\eqref{eq:nr_lin_eq} as:
\begin{equation}
\label{eq:nullspace_sum}
q^{k+1}=x^*+Zv^k,
\end{equation}
where $x^*$ is one of an infinite number of solutions for $A_{12}^Tx=d$ (eg.\ a least-squares solution for $d\neq 0$ would suffice) and $v^k\in \mathbb{R}^{n_l}$ is unknown. 
Substituting for $q^{k+1}$ in the first block row  of~\eqref{eq:nr_lin_eq} and pre-multiplying by $Z^T$ results in the smaller linear system
\begin{equation}
\begin{aligned}
Z^TF^{k}Z \; v^k&=Z^T[(F^k-G^k)q^k-A_{10}h_0-F^k x^*],\label{eq:nullspace_lineq}
\end{aligned}
\end{equation}
where $F^k=NA_{11}^k$ and $G^k=A_{11}^k.$

The heads are then calculated by solving
\begin{align}
A_{12}^TA_{12} \; h^{k+1} &= A_{12}^T\{(F^{k}-G^{k})q^k-A_{10}h_0-F^kq^{k+1}\}.
\label{eq:nullspace_lineq_a}
\end{align}
A null space algorithm based Newton method first solves for $x^*$ such that $A_{12}^Tx^*=d$, and
then iteratively solves~\eqref{eq:nullspace_lineq} and~\eqref{eq:nullspace_lineq_a} in sequence until convergence is achieved.  Of course,  \eqref{eq:nullspace_lineq_a} need only be solved when the iterates are near convergence because  the flow equations~\eqref{eq:nullspace_lineq} do not depend on the pressure heads; see Subsection~\ref{subsec:stopcriteria} for a discussion on convergence criteria.
The null space method has the following computationally advantageous  properties:
\begin{itemize}
 \item Where the null space dimension $n_l$ is small, the linear system in~\eqref{eq:nullspace_lineq} is  smaller than the Schur method equations~\eqref{eq:nr_lin_eq_1a}. Since $F^k$ is diagonal, the null space problem will be sparse if $Z$ is sparse. As will be shown in Table~\ref{table:Zproperties},  with an appropriate choice of $Z$, the number of non-zeros in  $Z^TF^{k}Z$ is much less than the number of non-zeros in $A_{12}^TF^{k}A_{12}$ for most WDN models.
 \item The matrices $Z^TF^{k}Z$ can be shown to be symmetric positive definite (SPD). Even when $F^k$ is singular, the condition ${\ker(F^k)\cap \ker{A_{12}^T}=\{0\}}$ is sufficient to show positive definiteness. 
  \item The matrix coefficient of~\eqref{eq:nullspace_lineq_a}, $A_{12}^TA_{12}$, is  similarly SPD -- see the appendix of~\cite{elhay2014reformulated} for proof that $A_{12}$ has full rank, and positive definiteness follows. Since this matrix depends only on the network topology and  does not change with Newton iterations or demand, a single numeric factorization can be used for multiple hydraulic analyses. 
  
  \item In extended time simulations, we need to solve for different $x^*$ as the demands $d$ vary. Now, since $x^*$ is in the range space of $A_{12}$,  let $x^*=A_{12}w,\; w\in\mathbb{R}^{n_n}$ and substituting for $x$ we get:
  \begin{equation}
   \begin{aligned}
  A_{12}^TA_{12}w&=d,\label{eq:nullspace_xstar}
\end{aligned}
  \end{equation}
Therefore, the same single factorization of the SPD system~\eqref{eq:nullspace_lineq_a} can be used to solve for $w$ by forward and back substitutions and $x^*\gets A_{12}w$).
  
 \item  Similarly, the matrix $Z$ is computed only once for multiple hydraulic simulations.  

\end{itemize}


For (sparse) linear solvers, since the matrix factorization stage is the most computationally demanding stage~\cite[Appx.\ C]{boyd2004convex}, the reuse of a single factorization for~\eqref{eq:nullspace_lineq_a} results in large computational savings.   It is also desirable that the condition number of $Z$ be low since the condition number of $Z^TF^{k}Z$ 
is bounded by its square.  Depending on the the method of choice for computing $Z$, a number of null space methods can be adopted; Algorithm~\ref{alg:null_full} shows the null space  Newton method tailored to demand-driven hydraulic analysis. 

\newcommand{\algrule}[1][.2pt]{\par\vskip.5\baselineskip\hrule height #1\par\vskip.5\baselineskip}
\begin{algorithm}[t]
{\bf Preprocessing for extended time simulations:} Compute all constants
\begin{enumerate}[label=(\roman{*})]
\item Compute null-space basis $Z$
\item  Factorize $A_{12}^T A_{12}$ (i.e.\ compute $L$ such that $LL^T=A_{12}^T A_{12}$)
\end{enumerate}	
{\bf Preprocessing for a given demand $d$:} 
\begin{enumerate}[label=(a)]
\item {Solve for $x^*$ from~\eqref{eq:nullspace_xstar}: $LL^Tw=d,\quad x^*\gets A_{12}w$}
\end{enumerate}	

{\bf Input: $\delta_N$, $k_{max}$, ($x^*$, $L$, $Z$) , $q^0,h^0$}

{\bf Algorithm:}
     \algrule
\begin{algorithmic}[1]
\STATE set $k=0$, and compute $G^0,\; F^0,\;\|f(q^0,h^0)\|_\infty$
\WHILE{$\|f(q^k,h^k)\|_\infty>\delta_N$ AND $k\leq k_{max}$ } 
\STATE $ F^k=Regularize(F^k)$ 
\STATE $\underbrace{Z^TF^{k}Z}_{X^{k}}=\sum\limits_{i=1}^{n_p}{f_i^{k} z_i z_i^T}$ \label{alg:lineZFZ}\\
\STATE Solve $ X^k v^k=b^k$
\STATE $ q^{k+1}=x+Zv^k$
\STATE Recompute $G^k, \; F^k$ $\quad$ \label{alg:lineHlossUpdate}
\STATE Solve $ LL^T h^{k+1}=b(q^{k+1}) $  \label{alg:lineLL}
\STATE Set $k$ to $ k+1$
\STATE Compute the Residual error $\|f(q^k,h^k)\|_\infty$ \label{alg:lineError}
\ENDWHILE
\end{algorithmic}
\caption{Exact Newton method with null space algorithm}
\label{alg:null_full}
\end{algorithm}

In~\cite{Boulos1993explicit}, their explicit loop method is shown to remain well posed (i.e.\ to have a unique solution) provided no loop contains all zero flows. 
For null space methods, it has been shown in~\cite{nielsen1989methods} that the problem stays well posed as long as none of the loops have zero flows in all pipes. Therefore, compared to a Schur method, a null space algorithm is more robust in dealing with the zero flow problem~\cite{elhay2014reformulated,abrahamnull2014}. However, it is quite usual to have badly conditioned hydraulic analysis problems when large scale operational networks are considered. For example, Figure~\ref{fig:FRhist} shows the distribution of the frictional loss coefficients for network BWKWnet and the elements of the diagonal matrix $G^k:=A_{11}^k$, corresponding to the  pipe flow solutions (from a null space algorithm) in Figure~\ref{fig:Flowhist}. The ratio of the largest to the smallest friction factors is of order $10^8$. When coupled with a large range for pipe flows, this results in even larger condition numbers for $G^k.$
\begin{figure}[h!]%
 \centering
     \subfloat[$r$ (flow resistance coefficients)]{ \includegraphics[height=0.25\textwidth]{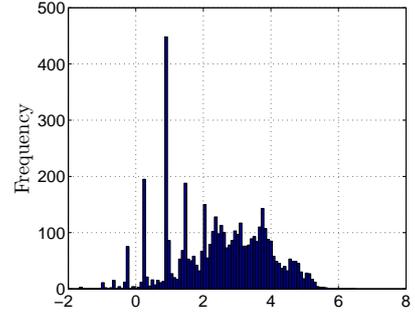}} \\%
     \subfloat[$\texttt{diag}(A_{11})$]{ \includegraphics[height=0.25\textwidth]{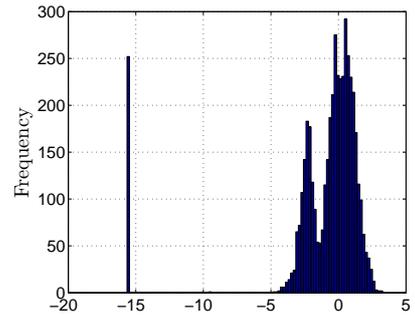}} %
    \caption{A histogram showing the distribution of frictional loss coefficients for the pipes and associated elements in the diagonal matrix $A_{11}$ in~\eqref{eq:nr_lin_eq} for the network BWKnet.}%
    \label{fig:FRhist}%
\end{figure} 
To avoid numerical ill conditioning and possible positive semidefiniteness of the GGA linear systems due to singular diagonal elements of $A_{11}$, zero and very small flows were replaced by arbitrary small positive number $\delta$ in~\cite{todini1988section}; zero flow cases are never allowed for in any link. However, as can be seen in~Figure~\ref{fig:FRhist}, even when zero flows are replaced by a small constant (for example, here we found that \ $\delta=10^{-6}$ was the smallest number that did not cause divergence in the Newton method), the condition number of $A_{11}$ is quite large (of the order $10^9$ here) resulting high condition numbers for $Z^TF^kZ$ and $A_{12}^TF^kA_{12}$. A rule of thumb implies a loss of a single decimal place in solving a linear system for every power of 10 increase in the condition number~\cite{elhay2011dealing}.  
For such systems, a systematic Jacobian regularization method is proposed in~\cite{elhay2011dealing} to restrict the condition number of the linear systems.   Using simple computations,  the work in~\cite{elhay2011dealing} suggests a systematic way to choose $T$ so that the condition number of the `regularised' matrix $\tilde{F}^k:=F^k+T^k$ is bounded above by some given number $\bar{\kappa}$, i.e.\ $\kappa_2(F^k+T^k)\leq \bar{\kappa}.$ Because $\tilde{F}^k$ is diagonal and invertible, it is straightforward to derive the bound on the 2-norm condition number ${\kappa_2(Z^T (\tilde{F}^k)^{-1} Z)\leq \kappa_2(\tilde{F}^k) \kappa(Z)^2}$~\cite{elhay2011dealing}, using the triangle inequality for the matrix norm. Therefore, by reducing $\kappa_2(\tilde{F}^k)$, we can reduce $\kappa_2(Z^T\tilde{F}^kZ).$ This approach reduces the loss of accuracy or convergence caused by inverting a badly conditioned Jacobian. In~\cite{abrahamnull2014}, we show that such a regularization results in an inexact Newton method, still retaining  local linear convergence properties.

\subsection{Computing  null space bases }

Compared to Schur methods,  null space algorithms are  advantageous for solving problems where $n_p-n_n$ is small and where the saddle point structure is present. In this paper, we are concerned with demand-driven analysis where the demand $d$ is constant  resulting in the saddle point structure of the Newton equations in~\eqref{eq:nr_lin_eq}. In leakage analysis, pressure driven models are used where the demand is a function of nodal pressures, i.e. $d:=d(h)$, see~\cite[Eq. (1)--(3)]{giustolisi2011demand}.   Since the derivative of the continuity equation with respect to pressure is nonzero in pressure-driven simulations, the $A_{22}$ block of the matrix on the left hand side of~\eqref{eq:nr_lin_eq} becomes non-zero;  the standard saddle-point structure is lost and makes the applicability of the null space algorithms limited to demand-driven cases. 

In addition to  demand-driven hydraulic analysis~\cite{nielsen1989methods,rahal1995cotree,elhay2014reformulated}, null space algorithms have been exploited in optimization,  electrical circuit analysis, computational structural mechanics, and unsteady fluid dynamics applications where problems have this saddle point structure; see~\cite[Sec.\ 6]{benzi2005numerical} for a large list of literature on such applications. In all these, {\it Kirchhoff's second law}  is exploited; it states that the energy difference (hydraulic head difference in our case) around a closed loop is zero. 

The Hardy Cross method~\cite{cross1936analysis} is in fact a null space method, although not reported as such at the time~\cite{ormsbee2006history}. 
Starting with an initial guess of flows that satisfy continuity of flow at all junctions, the method seeks flow corrections for each loop such that Kirchhoff's second law would be valid. The set of nonlinear equations in the flow corrections for each loop are solved iteratively by a first-order Taylor model until the conservation laws around all loops are met; all iterates satisfy flow continuity. In the computer era, the Hardy Cross method was extended to simultaneously solve all the loop flow corrections via the Newton method (often called ``simultaneous loop flows method'')~\cite{epp1970efficient}, improving the convergence properties of the original approach and making it fast enough for larger size networks. Note from~\eqref{eq:nullspace_sum} that a null space Newton method first finds a solution $x^*$ that satisfies flow continuity and, at each iteration,  computes adjustments $v$ in the kernel space of $A_{12}^T$ until  energy conservation is satisfied; this equivalence with the simultaneous loop flows method was made in~\cite{nielsen1989methods} and subsequent literature~\cite{rahal1995cotree,elhay2014reformulated}.   

By using the structure of the incidence matrix $A_{12}$, a number of methods that require no floating point operations can be employed to construct a sparse null basis $Z$ with desirable properties. For example,  if $Z$ is very sparse, a sparse $Z^TF^kZ$ can be explicitly formed for solution with direct methods even for large scale systems. 
In~\cite{epp1970efficient}, what they call a ``natural set of loops'' are used since they generate a low-bandwidth banded matrix $Z^TF^kZ$	 and so  reduce memory requirements in solving~\eqref{eq:nullspace_lineq}. An automatic loop numbering scheme is employed so as to generate an independent set of loops, i.e.\ ones that share the minimum number of links with other loops. 

In~\cite{rahal1995cotree}, a  graph-theoretic approach that is faster than the approach in~\cite{epp1970efficient} but requiring no floating-point arithmetic is used to generate fundamental basis that have similar memory requirements as the ``natural basis'' from~\cite{epp1970efficient}.
Using graph theoretic notation, let $\mathcal{G}(V,E)$ denote a connected, undirected graph of the water distribution network with $n_p$ edges and a set of ${n_n+n_0}$ vertices corresponding to unknown  and fixed head nodes. 
Although each link is endowed with an arbitrary fixed reference direction specifying the direction of flow, the graph is still undirected as the flow is allowed in both directions.  
Let $\mathcal{T}(V,E_1)$ denote a spanning tree of $\mathcal{G}$,  a sub-graph of $\mathcal{G}$ that contains  a subset of edges $E_1\subseteq E$ that  span all the vertices $V$  with no loops/cycles. The process employed in~\cite{rahal1995cotree} uses Kirchhoff's classical method, which  finds the null basis by using a spanning tree of the network and then constructing loops using the respective co-tree (i.e.\ the set {$E\setminus E_1$} )~\cite{benzi2005numerical}. An edge-loop matrix is formed by  adding a single chord from any of the co-tree edges, forming  loops in the process. For each such fundamental loop, a column of $Z$ is defined where the entry for each link in the loop is set to $\pm1$ depending on the direction of flow assigned in the incidence matrix.  Such fundamental basis will have full column rank since each loop in the basis contains at least one edge which is not contained in any other loop in the basis.
Loop equations are then solved in~\cite{rahal1995cotree}  to find flows in the co-tree chords, which are then used to update the spanning tree flows at convergence. 
The property of the matrix $Z$ will of course depend on the spanning tree used. For example, the sparsity of $Z$ will depend on the particular spanning tree used; the tree for which the sum of the number of edges in the fundamental loops  is minimized results in the sparsest basis $Z$. However, finding such a tree,  or generally the sparsest $Z$,  is an NP-hard problem~\cite{benzi2005numerical}. Nonetheless, practical heuristics exist for solving this problem approximately.

\begin{table*}[h!]
{
\begin{center}
\caption{  SPLU, RCTM and SPQR refer to the null basis generated using sparse LU~\cite{davis2004umfpack}, the matrix reduction method of~\cite{elhay2014reformulated} and sparse QR method~\cite{davis2011spqr}, respectively.  We denote the the number of non-zero rows of the matrix $Z$ by $|\mathcal{E}_2(Z)|$. }
\label{table:Zproperties}
\tabcolsep=0.15cm
\begin{tabular}{|l|l|l|l|l|l|l|l|l|l|l|l|}
\cline{1-11}
& \multicolumn{3}{c|}{$\kappa(Z^TZ)$} &  {\small $\kappa(A_{12}^TA_{12})$} &\multicolumn{3}{c|}{$\frac{nnz(Z^TF^kZ)}{nnz(A_{12}^TF^kA_{12})}$ \%}& \multicolumn{3}{c|}{$\frac{|\mathcal{E}_2(Z)|}{n_p}$ \%}\\
\cline{2-4,6-11}
Network & 	SPLU	&	RCTM 	&   SPQR& &SPLU&	RCTM & SPQR & SPLU &RCTM & SPQR \\
\cline{1-11}
CTnet	&$1.6\times10^2$&$1.8\times10^2$&1.0&$4.2\times10^3$	&25.6 &32.7 &150.1 &65.1&65.1& 69.8\\
\cline{1-11}
Richnet	&$1.4\times10^2$&$1.4\times10^2$&1.0&$3.0\times10^4$	&13.1 &13.2 &77.7  &48.7&48.7& 61.3\\
\cline{1-11}
WCnet	&$5.3\times10^2$&$7.2\times10^2$&1.0&$4.9\times10^4$	&43.3 &45.6 &318.3 &58.0&58.0& 58.4 \\
\cline{1-11}
BWFLnet	&$3.7\times10^2$&$2.0\times10^2$&1.0&$1.6\times10^5$	&9.1  &8.7  &43.4  &51.7&51.7& 53.1\\
\cline{1-11}
EXnet	&$3.8\times10^3$&$1.3\times10^3$&1.0&$1.0\times10^5$	&115.4&102.1&1147  &80.1&80.1&82.3\\
\cline{1-11}
BWKnet	&$7.7\times10^1$&$9.6\times10^1$&1.0&$2.0\times10^6$	&2.0  &2.3  &24.4  &30.6&30.6&39.3\\
\cline{1-11}
NYnet	&$1.0\times10^4$&$1.3\times10^4$&1.0&$1.6\times10^6$	&75.9 &73.2 &1154  &74.1&74.1&79.3\\
\cline{1-11}
\end{tabular}
\end{center}
}
\end{table*}

Unlike in~\cite{rahal1995cotree}, the methods of~\cite{nielsen1989methods} and~\cite{elhay2014reformulated} do not consider virtual-loops, spanning trees, and co-trees -- an algebraic approach is taken in forming the null bases.  Since the incidence matrix $A_{12}\in\mathbb{R}^{n_p\times n_n }$ has full column rank, it follows that there always exist  permutation matrices $P$ and $Q$ such that 
\begin{equation}\label{eq:permutetriang}
 QA_{12}^TP=\begin{bmatrix}L_{1}  & L_{2}\end{bmatrix}=:L,
\end{equation}
where $L_1\in\mathbb{R}^{n_n\times n_n}$ is invertible, and $L_2\in\mathbb{R}^{n_n\times n_l}$. A straightforward substitution shows that the matrix
\begin{equation}\label{eq:fundbasisZ}
 Z=P \begin{bmatrix}-L_{1}^{-1}L_{2}\\ I_{n_l}\end{bmatrix}
\end{equation}
is a null basis for $A_{12}^T,$ i.e.\ $A_{12}^TZ=0$~\cite{benzi2005numerical}. Such a basis is called a fundamental basis~\cite{benzi2005numerical} and can be formed in many ways. 

In~\cite{nielsen1989methods}, no assumptions are made on the factorization~\eqref{eq:permutetriang} but that  $L_{1}$ be invertible and $Q=I_{n_n}$. In the formulation of~\cite{elhay2014reformulated},  also called a reformulated co-tree flows method (RCTM), a simple matrix reduction based approach is proposed for null basis generation. In~\cite{elhay2014reformulated}, it is noted that all WDNs have at least one fixed head node  (eg.\ a reservoir or tank) connected to an unknown head node.  For such a link connecting the fixed head node to the unknown head node, the corresponding row of the $A_{12}$ matrix will have only one non-zero element. This non-zero element is used as an initial pivot in interchanging rows and columns. The permutations are repeated $n_n$ times to find row and column permutations $P$ and $Q$, respectively, resulting in a lower triangular $L_1^T.$ A Gaussian substitution is then used to form the null basis~\eqref{eq:fundbasisZ}. In practice, this method results in very sparse and well conditioned null basis from the sparse matrices $L_1$ and $L_2$.

If we consider a triangular structure for~\eqref{eq:permutetriang} similarly to~\cite{elhay2014reformulated}, a well-conditioned null space basis can be computed from a sparse LU factorization; this is successfully used in the reduced Hessian methods of the SQP package SNOPT~\cite{gill2002snopt}. Let
\begin{equation}\label{eq:luZ}
 P^TA_{12}Q=LU,
\end{equation}
where $L^T=\begin{bmatrix}L_{1}  & L_{2}\end{bmatrix}$, $L_1\in\mathbb{R}^{n_n\times n_n}$ is lower-triangular with a unit diagonal,  $ U\in\mathbb{R}^{n_n\times n_n}$ is upper triangular, $L_2\in\mathbb{R}^{n_n\times (n_p-n_n)}$ and $Z$ is as in~\eqref{eq:fundbasisZ}. 
To compute $L$ and $P$, we use the sparse package UMFPACK~\cite{davis2004umfpack}, a state-of-the-art ANSI C library of routines for solving sparse linear systems via the LU factorization, which is also the \texttt{LU} function in MATLAB. We chose this package because, in addition to being one of the fastest packages for general sparse unsymmetric LU factorization problems, UMFPACK has also been shown to  produce the sparsest LU factors for  circuit simulation problems~\cite{davis2010klu}. For sparse unsymmetric matrices, to which $A_{12}$ belongs, UMFPACK uses a column pre-ordering (COLAMD~\cite{amestoy2004algorithm})  to preserve sparsity. Partial pivoting is used to limit fill-in  and to improve numerical accuracy in the Gaussian elimination~\cite{davis2004umfpack}. 

The best conditioned null space basis should theoretically be orthonormal, and these can be computed using a QR factorization.
Every full rank matrix $A_{12}\in\mathbb{C}^{n_p\times n_n},n_p \geq n_n $ has a full QR factorization
 \begin{equation}
A_{12}=\begin{bmatrix}Q_1& Q_2\end{bmatrix}\begin{bmatrix}R\\ 0\end{bmatrix},\nonumber
\end{equation} 
where $Q=\begin{bmatrix}Q_1& Q_2\end{bmatrix}\in\mathbb{C}^{n_p\times n_p}$ is unitary and ${R\in\mathbb{C}^{n_n\times n_n}}$ is upper triangular.  Moreover, the factorization ${A_{12}=Q_1R,} $ with $R_{ii}>0$ is the unique Cholesky factor of $A_{12}^TA_{12}$~\cite[Sec.\ 5.2.6]{golub1996matrix}. 
Since the columns of  $Q_2$ span $\ker(A_{12}),$  we have $Z=Q_2$ such that $\kappa_2(Z'Z)=1.$  
In principle, the QR factorization also produces the Cholesky factor of $A_{12}^TA_{12}$ and so seems attractive. However, even with the sparsest QR factorizations (eg. SPQR~\cite{davis2011spqr}, a high performance multifrontal routine for calculating sparse QR factors of large sparse matrices), the bases are much more dense than those from an LU factorization. 

Table~\ref{table:Zproperties} presents some relevant properties of the null space bases generated via the three methods discussed, the case study networks are shown in Section~\ref{sec:caseStudy}. For each network considered, SPQR generates null bases that are numerically orthonormal, the best conditioned matrices possible.  
The matrices $Z^TZ$ from the  sparse LU (SPLU), and RCTM bases  have similar condition numbers with each other,  and are also  better conditioned compared to their corresponding Schur system matrices $A_{12}^TA_{12}.$ 

Since the computational cost of solving the null space and GGA linear systems and the storage required  depend on the sparsity of $Z^TF^kZ$  and $A_{12}^TF^kA_{12}$, respectively, we analyse the relative sparsity of these matrices in the next three columns. For most of the networks, the ratio of the number of non-zero elements in $Z^TF^kZ$ to the number of non-zero elements in $A_{12}^TF^kA_{12}$ is much smaller when the SPLU or RCTM basis are used, the smallest being 2.0\% for SPLU  applied to the network BWKnet.  As also noted in~\cite{elhay2014reformulated}, this reduced sparsity implies smaller memory requirements for the linear solves compared to  Schur methods, allowing bigger networks to be analysed on the same hardware resources. On the other hand, the null bases from SPQR are about an order denser than the ones from SPLU and RCTM in all the examples. Moreover, except for the sparsest of the networks (BWKnet, BWFLnet and Richnet), SPQR results in null space linear systems with bigger memory requirements than even the Schur method; the largest by about a factor of 11.5 for NYnet. For this reason,  we do not propose the use of a QR based basis in the null space algorithm. Instead, we propose the use of the  SPLU or RCTM from~\cite{elhay2014reformulated} for the computation of well conditioned and sparse null basis; in the rest of this article we adopt sparse LU generated basis in all our implementations.
%
\setlength\figureheight{0.25\textheight}
\setlength\figurewidth{0.7\textwidth}
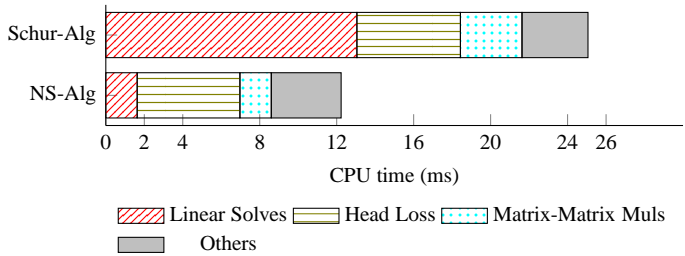
\begin{figure}[]
 \centering
\pgfplotstableread{ 
Label   First   Second  Third   Fourth
{NS-Alg}      1.63	 	5.34 	1.63    3.63
{Schur-Alg}    13.05    5.38    3.20    3.43	
}\datatable

\begin{tikzpicture}
\begin{axis}[
    xbar stacked,
    xlabel near ticks,
     xlabel={CPU time (ms)},
    xlabel style={anchor=near xticklabel,at={(xticklabel cs:0.5)},name=myxlabel},
    legend style={
        legend columns=3,
        at={(myxlabel.south)},
        anchor=north,
        draw=none
    },
    ytick=data,
    axis y line*=left,
    axis x line*=bottom,
    tick label style={font=\footnotesize},
    legend style={font=\footnotesize},    
    label style={font=\footnotesize},     
    xtick={0,2,4,8,12,16,20,24,26},%
    width=0.5\textwidth,
    bar width=6mm,
    yticklabels from table={\datatable}{Label},
    xmin=0,
     xmax=30,
    y=8mm,
    enlarge y limits=0.5,
    point meta=explicit,
    nodes near coords align=left
]
\addplot [draw=black, pattern color = red, pattern = north east lines,area legend] table [x=First, y expr=\coordindex] {\datatable}; 
\addplot [draw=black, pattern color = red!50!green, pattern = horizontal lines,area legend]table [x=Second, y expr=\coordindex] {\datatable};
\addplot [draw=black, pattern color = cyan, pattern = dots, area legend] table [x=Third, y expr=\coordindex] {\datatable};
\addplot [draw=black,fill=lightgray, area legend] table [x=Fourth, y expr=\coordindex] {\datatable};

\legend{Linear Solves,Head Loss,Matrix-Matrix Muls,Others};

\end{axis}
\end{tikzpicture}
\caption{Average CPU times (in ms) for Null space and Schur  algorithms; BWKWnet network.}
\label{fig:schurvsNullCPU_1}
\end{figure}

Figure~\ref{fig:schurvsNullCPU_1} shows a comparison of the computational cost of a Schur (an efficient regularised implementation of GGA)
and null-space Newton algorithms using the example  network BWFLnet.  In addition, similar to the analysis done in~\cite{guidolin2013using} for the  GGA method, Figure~\ref{fig:schurvsNullCPU_1}  details the  main computational blocks of both the Schur and null space algorithms. 
The  contribution of each block to the total computational time is shown.  
It is apparent that the matrix-matrix multiplications for the linear solves, the linear solves, and the head-loss computations together constitute over 75\% of the computational time. The ``Others'' block includes  matrix-vector multiplications, residual error norm  computations, Jacobian regularizations, and  diagonal matrix inversions in the case of the Schur method, which can add up to a significant portion of the total CPU time. 

The solve time for the linear systems of the null-space Newton method is much smaller than that of the Schur method, although two linear systems are solved in the former. The bigger of the two linear systems solved by the null space algorithm,~\eqref{eq:nullspace_lineq_a},  requires only a single factorization; the  computational cost in solving  these head equations is  partially amortised by the fact that a single numerical factorization of $A_{12}^TA_{12}$ is reused in a large number of simulations as long as the system connectivity remains constant. 
In addition, although the system matrix in~\eqref{eq:nullspace_lineq} changes at each Newton iteration,  it often has a significantly smaller fraction of nonzeros than that of the Schur matrix in~\eqref{eq:nr_lin_eq_1a}; see columns 6--8 of Table~\ref{table:Zproperties}. Since these SPD matrices are solved using triangular Cholesky factorization followed by backward and forward substitutions,  the factorization and substitution steps are roughly proportional to the sparsity factors~\cite[Appx.\ C.3.2]{boyd2004convex}.

 Sparsity structures of Hessians are often exploited in many nonlinear optimization problems to cheaply compute Newton steps~\cite[Ch.\ 9.7.2]{boyd2004convex}. To solve a linear system $Ax=b,$ where $A$ is SPD and sparse, a sparse Cholesky factorization followed by forward and back substitutions is used. Such a method computes a permutation matrix $P$ and a sparse lower triangular factor $L$ such that $P^TAP=LL^T.$ Matrix reordering algorithms (eg.\ AMD~\cite{amestoy2004algorithm} is used here) are used to compute the permutation matrices so as to reduce fill, i.e.\ the number of nonzeros in the factors $L$ with corresponding zeros in the matrix $A.$ Since the pattern of nonzeros and fill in the factors mostly depend only on the pattern of nonzeros in $A$ and not on  numerical values of $A$,  matrix factorization is divided into two steps -- {\itshape symbolic factorization} (i.e.\ determining the reordering matrices $P$ and the non-zero patterns of $L$) followed by  {\itshape numerical factorization}, where the non-zero numerical values of $L$ are computed. Similar to the Hessians in~\cite[Ch.\ 9.7.2]{boyd2004convex}, \eqref{eq:nullspace_lineq} has a matrix with constant sparsity pattern although its values change at each Newton iteration. Therefore, the cost of factorization can be partially amortised by using a single symbolic factorization for all numeric factors across multiple Newton steps and extended time simulations. 

Table~\ref{table:Singlesymfact} shows the CPU time reductions gained from reusing symbolic factors, where the computation times are averaged over multiple extended time simulations (1000 simulations). The results were generated using a hydraulic solver implemented in C++  using the Eigen library of numerical solvers~\cite{jacob2012eigen} -- see Section~\ref{sec:caseStudy} for network model and implementation details. For these examples, the CPU time is on average reduced by about 1-16\%.  From Table~\ref{table:Zproperties} and the description of the networks in Table~\ref{table:CSnetworks}, we note that the more looped networks (i.e.\ the ones with higher degree) result in larger and denser linear system matrices $X^k:=Z^TF^kZ,$ which then require more flops to solve.  On the other hand, the size and density of the null space linear systems decrease with the degree of the network, making smaller the contribution of the linear solves to overall computational cost -- compare EXnet with BWFLnet and BWKWnet for example networks. 
The results in Table~\ref{table:Singlesymfact} reflect this; the reusing of symbolic factors in the linear solvers has the most impact for the denser networks, the highest  being around 15\% for EXnet, whcih is the densest of the network models. For the sparsest network, BWKWnet, the CPU time savings are  smallest at approximately 1\%.

\begin{table}[!t]
{\small
\caption{ A comparison of CPU times for the extended time simulations with Algorithm~\ref{alg:null_full}  by re-using a single symbolic factorization for $X^k$ ($t_1$), and without doing so ($t_0$). }
\label{table:Singlesymfact}
\begin{center}
\begin{tabular}{|l|l|l|}
\cline{1-3}
Network & $t_1$ (ms)  & $t_1/t_0$\\
\cline{1-3}
CTnet   & 3.64	& 1.098\\
\cline{1-3}
Richnet &  10.75	& 1.088\\
\cline{1-3}
WCnet   &  28.7	0& 1.093\\
\cline{1-3}
BWFLnet	&11.6 & 1.037\\
\cline{1-3}
EXnet	&  30.65	& 1.155\\
\cline{1-3}
BWKWnet	&  17.27& 1.011\\
\cline{1-3}
NYnet	& 535.67 & 1.147\\
\cline{1-3}
\end{tabular}
\end{center}
}
\end{table}

In~\cite{guidolin2013using},    head loss computations are shown to contribute significant  computational overheads; Figure~\ref{fig:schurvsNullCPU_1} also demonstrates this to be the case for both the Schur and null space methods.  Data parallel high performance computing techniques are analysed and used in~\cite{guidolin2013using} to accelerate pipe head loss computations at each linear solve of a GGA iteration.  
In the next section, we propose a   partial update scheme  to reduce the computational cost associated with head loss computations, while maintaining the data parallelism (, i.e.\ for each pipe, a head loss computation is dependent only on the flow and roughness characters of the same pipe). We will define the partial update scheme and present its convergence analysis. Based on the partial updates, we also propose a stopping criteria heuristic for the null space method, which will reduce the number of nodal head computations in~\eqref{eq:nullspace_lineq_a} or line~\ref{alg:lineLL} of Algorithm~\ref{alg:null_full}.

%
%
%
%
\section{Partial update method for the null space algorithm}\label{sec:parUpdate}
\subsection{Algorithm derivation and convergence}

Lets reconsider the fundamental null space basis in~\eqref{eq:fundbasisZ} again, where $\mathcal{T}(V,E_1)$ denotes a spanning tree of the network graph $\mathcal{G}(V,E)$. Each column of $Z$ defines a fundamental cycle, which contains a chord from the set of co-tree edges ($E\setminus E_1$) and  a corresponding unique path in the spanning tree $\mathcal{T}(V,E_1)$ connecting the two nodes incident on the chord. Let $E_2\subset E$ represent the union of edges in all such fundamental cycles, i.e.\ the set of pipes involved in the loop equations. 
Then, the cardinality of the set $E_2$ equals the number of rows of the matrix $Z\in\mathbb{R}^{n_p\times n_l}$ that are not identically zero. 
If we consider the linear system~\eqref{eq:nullspace_lineq}, we can rewrite the matrix $X^k:=Z^TF^kZ$  (i.e.\ line~\ref{alg:lineZFZ} of Algorithm~\ref{alg:null_full}) as
\begin{equation}\label{eq:XkFormula}
X^k=\sum\limits_{i=1}^{n_p}{f_i^{k} z_i z_i^T},
\end{equation}
where $f_i^{k} $ is the $i^{th}$ diagonal element of the diagonal matrix $F^k$ and $z_i^T$ is the $i^{\text{th}}$ row of $Z^T.$ Let $\mathcal{E}_2$ be the index set of pipes  belonging to the set $E_2.$ Then,~\eqref{eq:XkFormula} reduces to 
\begin{equation}\label{eq:XkFormula1}
X^k=\sum\limits_{i\in\mathcal{E}_2}{f_i^{k} z_i z_i^T},
\end{equation}
because the rows of $Z$ that are identically zero have no contribution. In other words, flow updates at each iteration of the null space Newton method will not involve the pipes not belonging to $E_2$. 

 \begin{figure}[tbh!]%
    \centering
\includegraphics[width=0.5\textwidth]{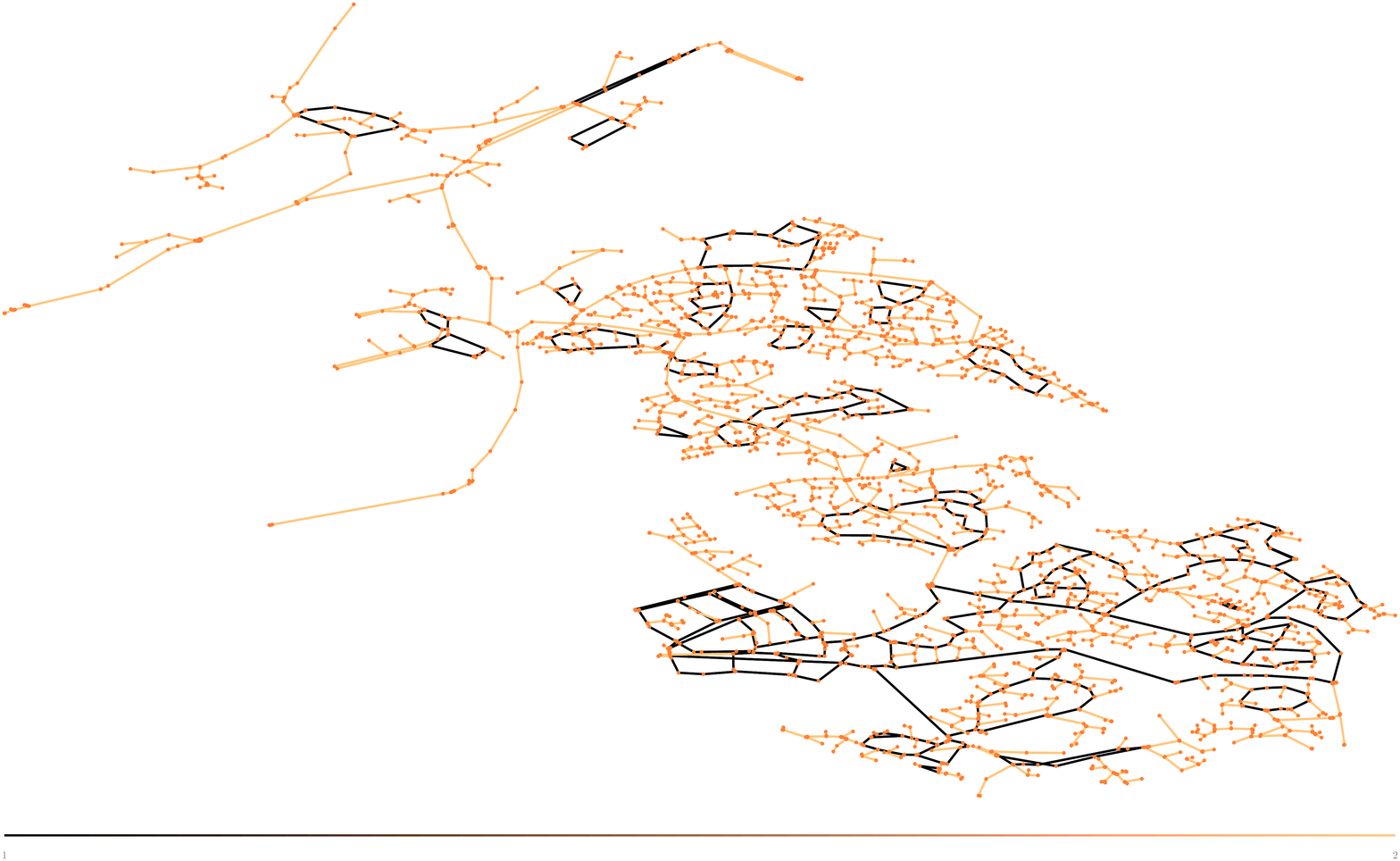}
    \caption{Proportion of links involved in flow updates are shown using black edges for the network BWKWnet.}%
\label{fig:fracloopflows} 
\end{figure} 
 To compute each Newton step in~\eqref{eq:nr_lin_eq}, the Jacobian $\nabla f(\cdot)$ is re-evaluated at  each flow iterate  by computing the associated frictional headlosses.  Figure~\ref{fig:fracloopflows} shows, in black, the $30.6\%$ of all pipes that form $E_2$ for the network BWKWnet. 
The last three columns of Table~\ref{table:Zproperties} also show the fraction of pipes involved in the loop equations when  three different methods are used to generate the null bases $Z.$ We note that this fraction can be as small as only 30\% of all links for the sparsest network. When using the SPLU and RCTM null bases, across the example networks, from 20\% to 70\% of the pipes do not belong to $E_2$, and so have flows that do change at each Newton iteration.
Therefore, we propose that the head loss equations be updated only for the set of pipes belonging to the set $E_2$, reducing computational cost significantly. 

 \begin{figure*}[tbh!]%
    \centering
       \subfloat[$k=1$]{\includegraphics[width=0.5\textwidth]{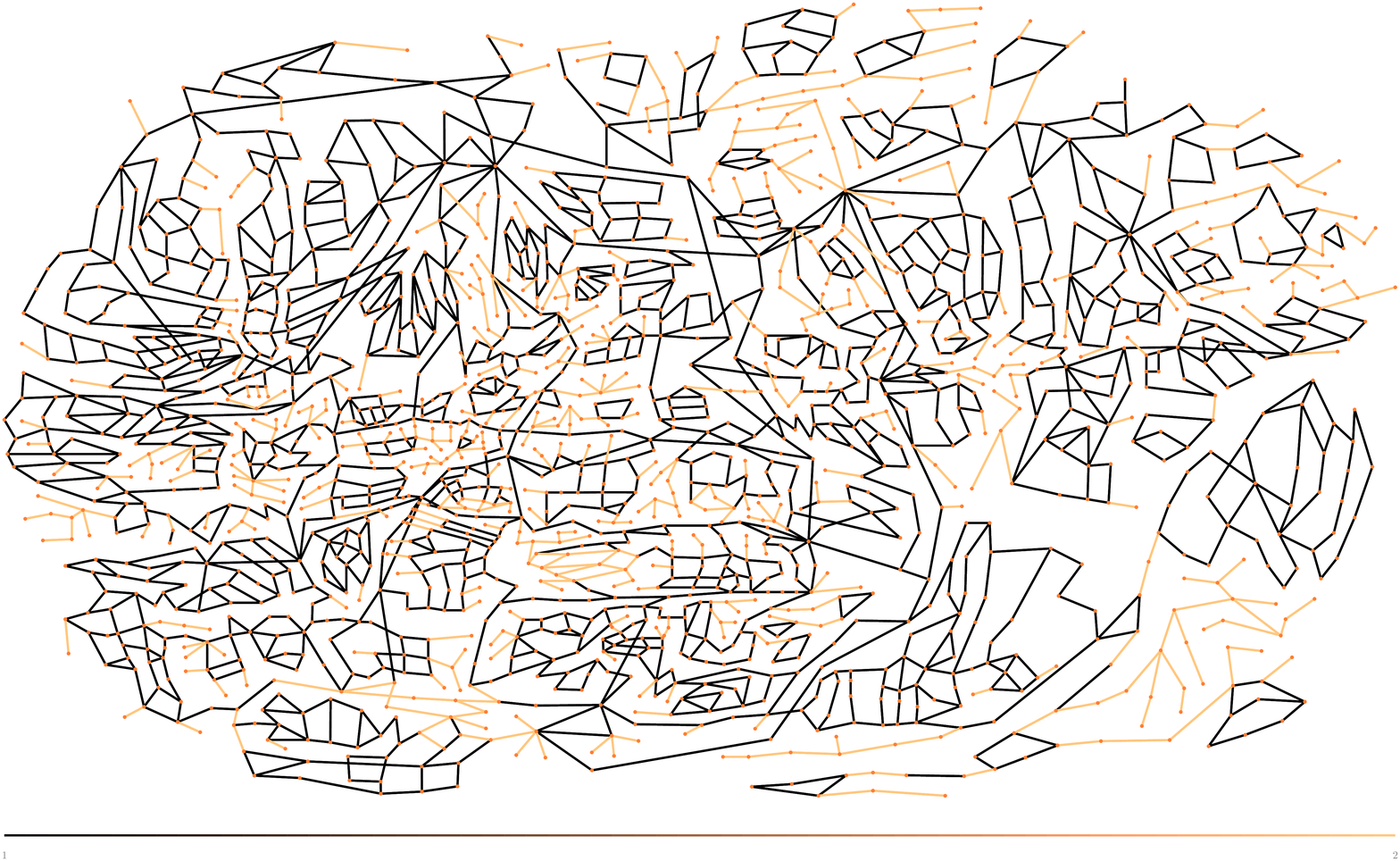}\label{fig:flowk_1} }
    \subfloat[$k=10$]{\includegraphics[width=0.5\textwidth]{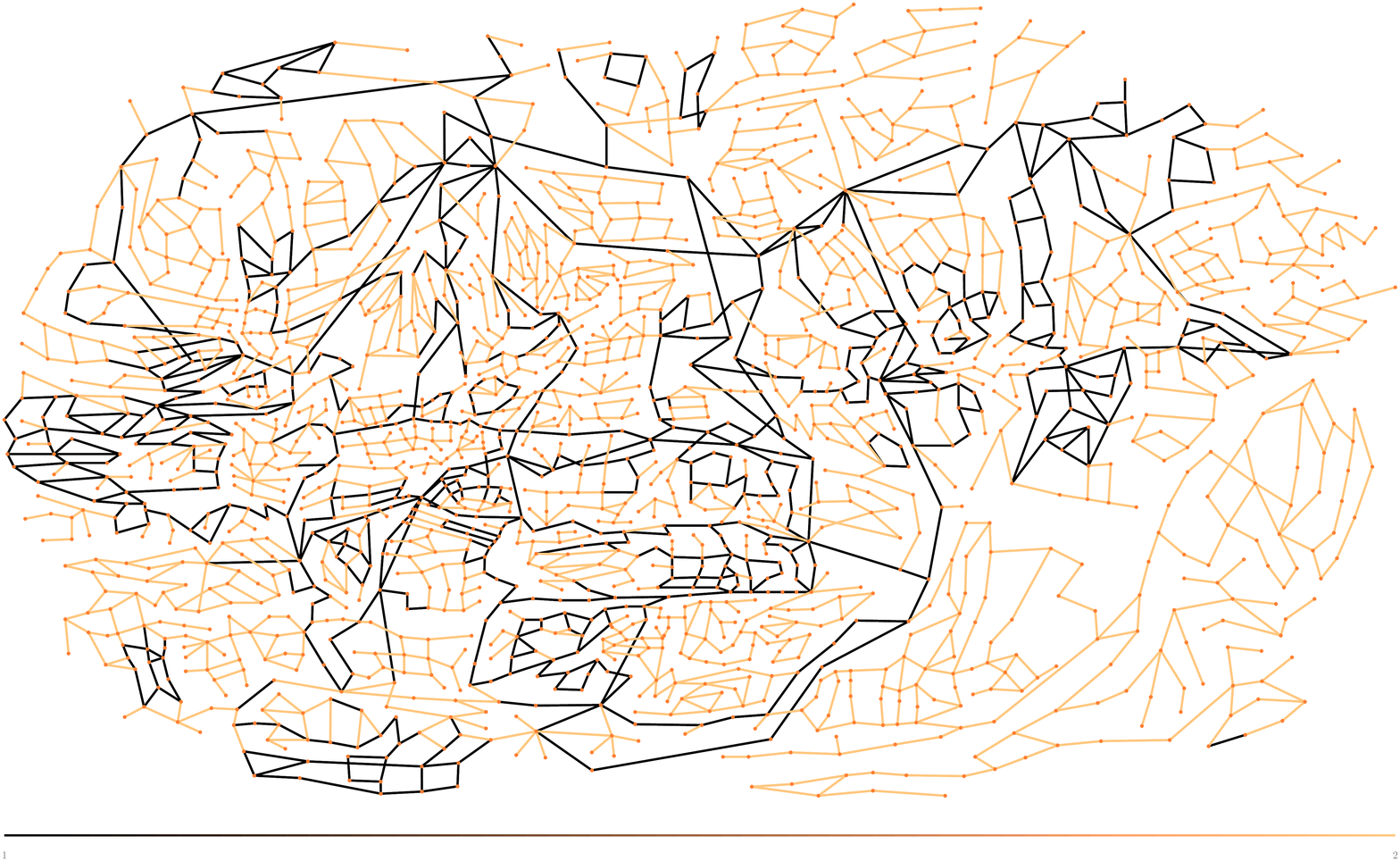} \label{fig:flowk_2}}
    \caption{The convergence of flows at Newton iteration $k$. The black edges show pipes whose flow values have  not yet converged at iteration $k.$}%
    \label{fig:pcg}%
\end{figure*} 
\begin{figure}[!bt]
 \centering
\includegraphics[width=0.45\textwidth]{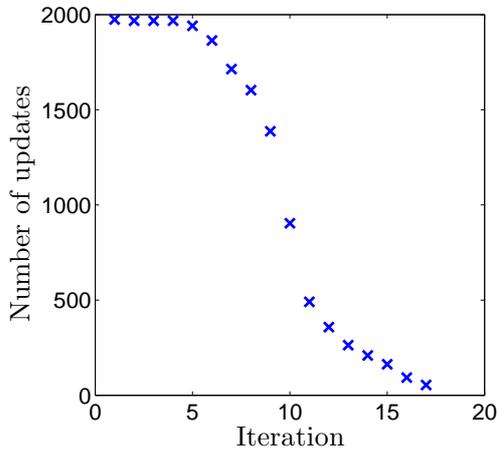}
\caption{The number of link flows that change by more than $1e^{-9}=\delta_N*\epsilon$ at each Newton iteration;  $\delta_N=1e^{-6},$  $\epsilon=1e^{-3}$ and  the EXnet network model, with $n_p=2465$ is used.}
\label{fig:numUpdates}
\end{figure}

The  plot in Figure~\ref{fig:flowk_1} shows (in black) the 1974 loop flows in $\mathcal{E}_2$ for the  network model EXnet, which roughly consists of $80\%$ of the links.  Here we also study the convergence of these flows; for example,   Figure~\ref{fig:flowk_2}  shows the  fraction of flows in  $\mathcal{E}_2$ that `have not converged' by the $10^{th}$ Newton iteration (, i.e.\ a flow has converged in the sense that it  does not change by more than a  small number at the given iteration, here  ${1e^{-9} \text{ m}^3/\text{s}}$). We propose that further computational savings can be made by updating the head losses ($G^k$ and $F^k$) only for flows that have not converged in this sense. This would reduce the number of flops required for these operations and so reduce overall computational complexity further.  For the  network model EXnet, Figure~\ref{fig:numUpdates} shows the number of flows that `have not converged'  at each iteration. Although the set of links in $E_2$ is a much bigger fraction ($\approx  80\%$)  of all links for the EXnet model, approximately half of the computationally demanding headloss computations can be avoided by  considering updates for only unconverged flows.

We introduce the concept of a partial update set here. Let the residual error tolerance for the Newton iterations be $\delta_N $ and let $0<\epsilon<1$ be a small number.  We define the  (partial) update set at the $k^\text{th}$ iteration as
\begin{equation}\label{eq:parupdateset}
\mathcal{U}^k(\epsilon,\delta_N):=\{i\in\mathcal{E}_2:|s_i^k|:=|q_i^{k+1}-q_i^k|\geq \epsilon \delta_N \}, 
\end{equation}
where $s_i^k$ is the Newton step in the flow update of the $i^\text{th}$ pipe or link at iteration $k$.  
At each Newton iteration, we  need not compute all the frictional headlosses across the network of links $i=1,\ldots,n_p$; the partial update formula recalculates headlosses only for the smaller set $\mathcal{U}^k(\cdot)$ as
\begin{equation}\label{eq:parupdateG}
G^{k+1}_{ii} = r_i |q^{k+1}_i|^{n_i-1},\;  f^{k+1}_{i}  = n_i G^{k+1}_{ii},\; \forall i\in \mathcal{U}^k(\epsilon,\delta_N), 
\end{equation}
for a HW model, and similarly for a DW model. Moreover, this results in $X^k:=Z^TF^kZ$ (on line~\ref{alg:lineZFZ} of Algorithm~\ref{alg:null_full})  to be only partially updated satisfying the following update formula:
\begin{equation}\label{eq:parupdateX}
X^k=X^{k-1}+\sum\limits_{i\in\mathcal{U}^k(\cdot)}{(f_i^{k}-f_i^{k-1}) z_i z_i^T},
\end{equation}
where $z_i$ is the $i^{\text{th}}$ column of $Z^T.$

In an exact Newton method for solving nonlinear equations $f(x)=0$, the linear systems $\nabla f(x^k)s^k=-f(x^k)$ are solved to find the Newton steps  $s^{k}:=x^{k+1}-x^k$ at each iteration. By Proposition~\ref{prop:ENM}, Algorithm~\ref{alg:null_full} is a Newton method for the hydraulic equations in~\eqref{eq:hydro_nl_eqn}. If Algorithm~\ref{alg:null_full} is coupled with the partial update formulae~\eqref{eq:parupdateG} and~\eqref{eq:parupdateX}, we introduce errors to both the Jacobian $\nabla f(x^k)$ and the right-hand side vector $f(x^k)$; an approximate linear system is solved and therefore the solution is `inexact'. In~\cite{abrahamnull2014}, we prove that there always exist update parameters $\epsilon$ that guarantee this inexact Newton method stays convergent. Since this proof is  outside the scope of the present paper, we state the claim here and investigate the parameters of the update set using simulations. 

\begin{proposition}(Partial-Updates Inexact Newton Method)
Assume the Newton method of Algorithm~\ref{alg:null_full} with error tolerance $\delta_N$ is coupled  with  the partial update formulae for the head losses as in~\eqref{eq:parupdateG}. Then, with the mild assumption that flows that have converged do not move away from the solution, there always exists a sufficiently small $\epsilon >0$, and update set $\mathcal{U}^k(\epsilon)$, such that the partial-update Newton scheme is an inexact Newton method, guaranteeing at least q-linear local convergence.  
\end{proposition}
\begin{proof}\nonumber
See~\cite[Proposition 1]{abrahamnull2014}.  
\end{proof}

For each of the networks in the case study, an extended simulation with 96 time steps was performed for the range of  partial update parameters  ${\log_{10}(\epsilon)=\begin{bmatrix}-9:0.5:-1\end{bmatrix}}$ and Newton tolerances ${\log_{10}(\delta_N)=\begin{bmatrix}-9:0.5:-3\end{bmatrix}}.$ Figure~\ref{fig:epsDelta1} and Figure~\ref{fig:epsDelta2} show a sweep of these parameter values for the Networks BWKWnet and EXnet, respectively. Considering the accuracy of the solution, i.e.\ the residual norm of the nonlinear equation at the solution $\|f(x)\|_\infty$, to depend on $\epsilon$ and $\delta_N$, we  plot its contours in Figures~\ref{fig:epsDelta1b} and Figure~\ref{fig:epsDelta2b} (with the maximum number of Newton iterations allowed $k_{max},$ set to 100 here). 
As $\epsilon \to 0$, the partial-update inexact Newton method approaches the original exact Newton method. In  Figure~\ref{fig:epsDelta1}, for $\epsilon \leq 1e^{-2.5}$, the inexact Newton method takes the same number of iterations as the exact method while satisfying the required level of error tolerance. For the example in  Figure~\ref{fig:epsDelta2}, $\epsilon < 1e^{-3}$ is sufficient. 

If $\epsilon $ is too large,  the algorithm with partial updates  either takes more iterations to attain the same level of accuracy in the solution or the required tolerance cannot be met within  $k_{max}$ iterations because the inexact Newton steps become significantly different to the steps of the exact Newton method.  Moreover,  if $\epsilon$ is not sufficiently small, the accuracy of the solution from the partial update method may not be within  the required Newton  tolerance  since the  residual error norm computations with the partially updated matrices $G^k$ would be far  from the true values. Similarly to the results in Figures~\ref{fig:epsDelta1} and~\ref{fig:epsDelta2} , we found that an epsilon value of  $1e^{-3}$ is sufficiently small and works well for all models considered under different error tolerances. 
Due to space limitations, we have shown only  EXnet and BWKWnet here because they lie on the opposite extremes of our case study models when considering their average degree or `loopedness'.

 \begin{figure*}[!tbh]%
    \centering
       \subfloat[]{\includegraphics[width=0.4\textwidth]{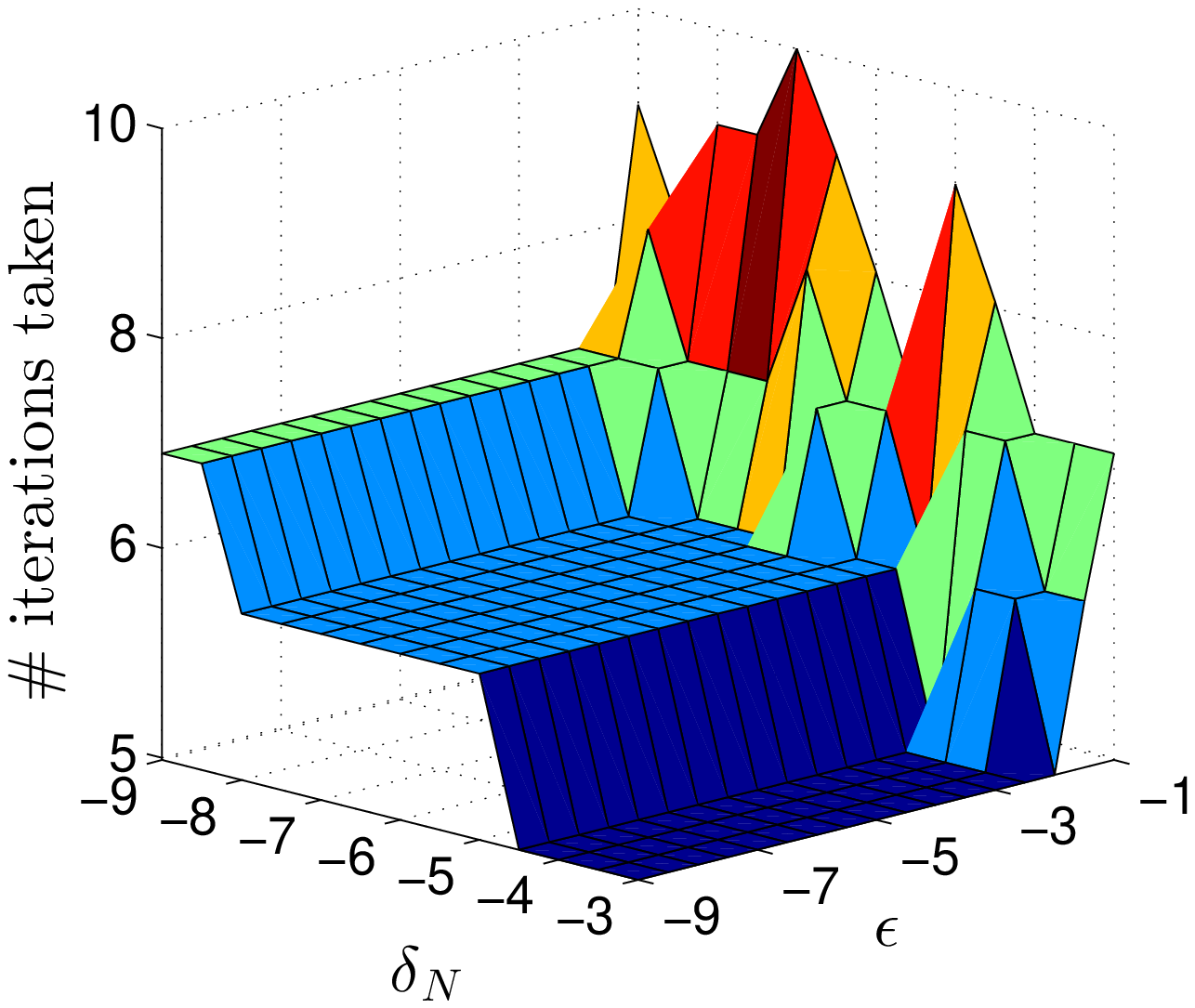}\label{fig:epsDelta1a} }
    \subfloat[]{\includegraphics[width=0.4\textwidth]{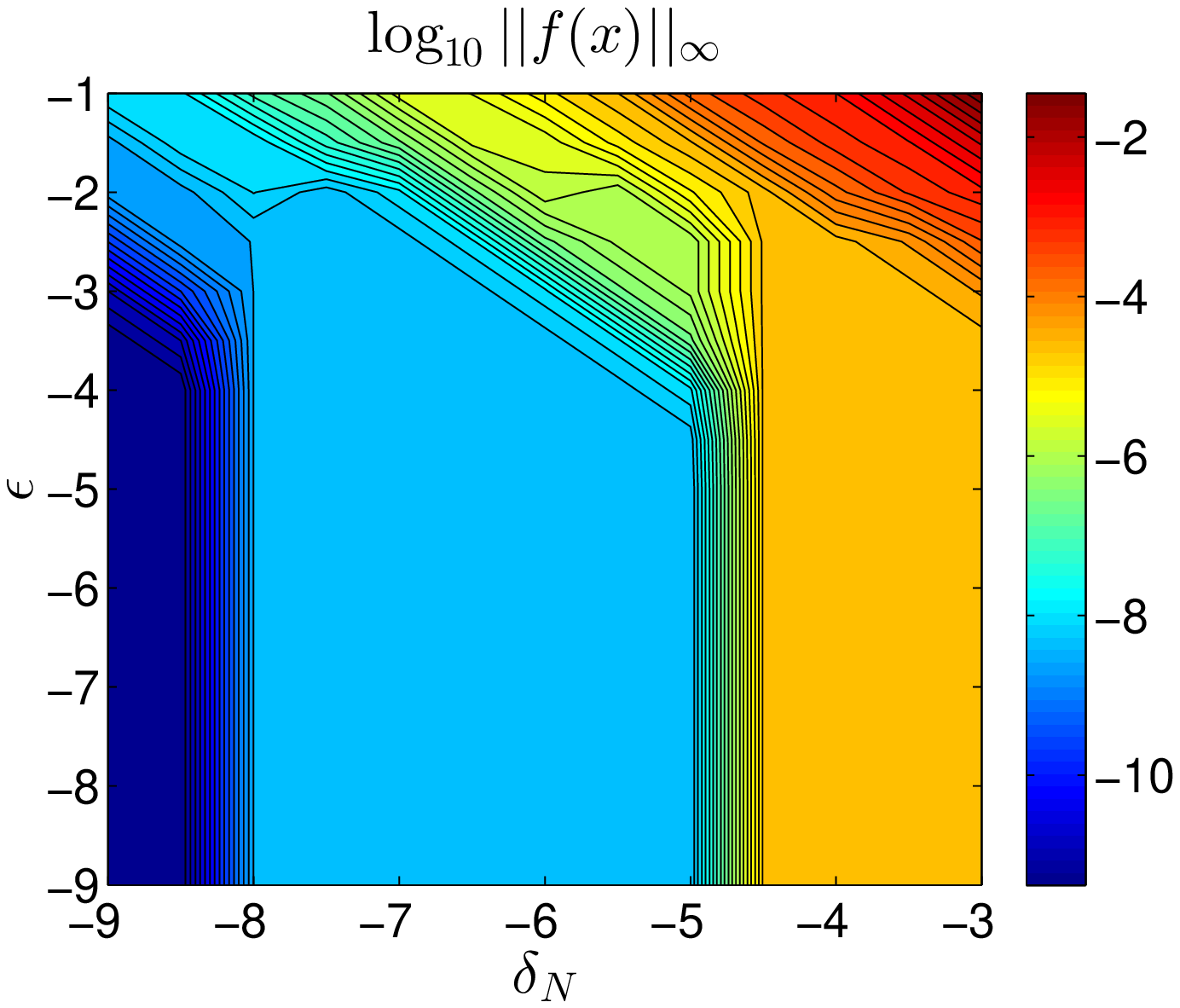} \label{fig:epsDelta1b}}
    \caption{A  parameter sweep for $\epsilon$ and $\delta_N$  against (a) the number of iterations (b) residual norm of the nonlinear equation, $\|f(x)\|_\infty$, for network BWKWnet.  }%
    \label{fig:epsDelta1}%
\end{figure*} 
 \begin{figure*}[tbh!]%
    \centering
     \subfloat[]{\includegraphics[width=0.4\textwidth]{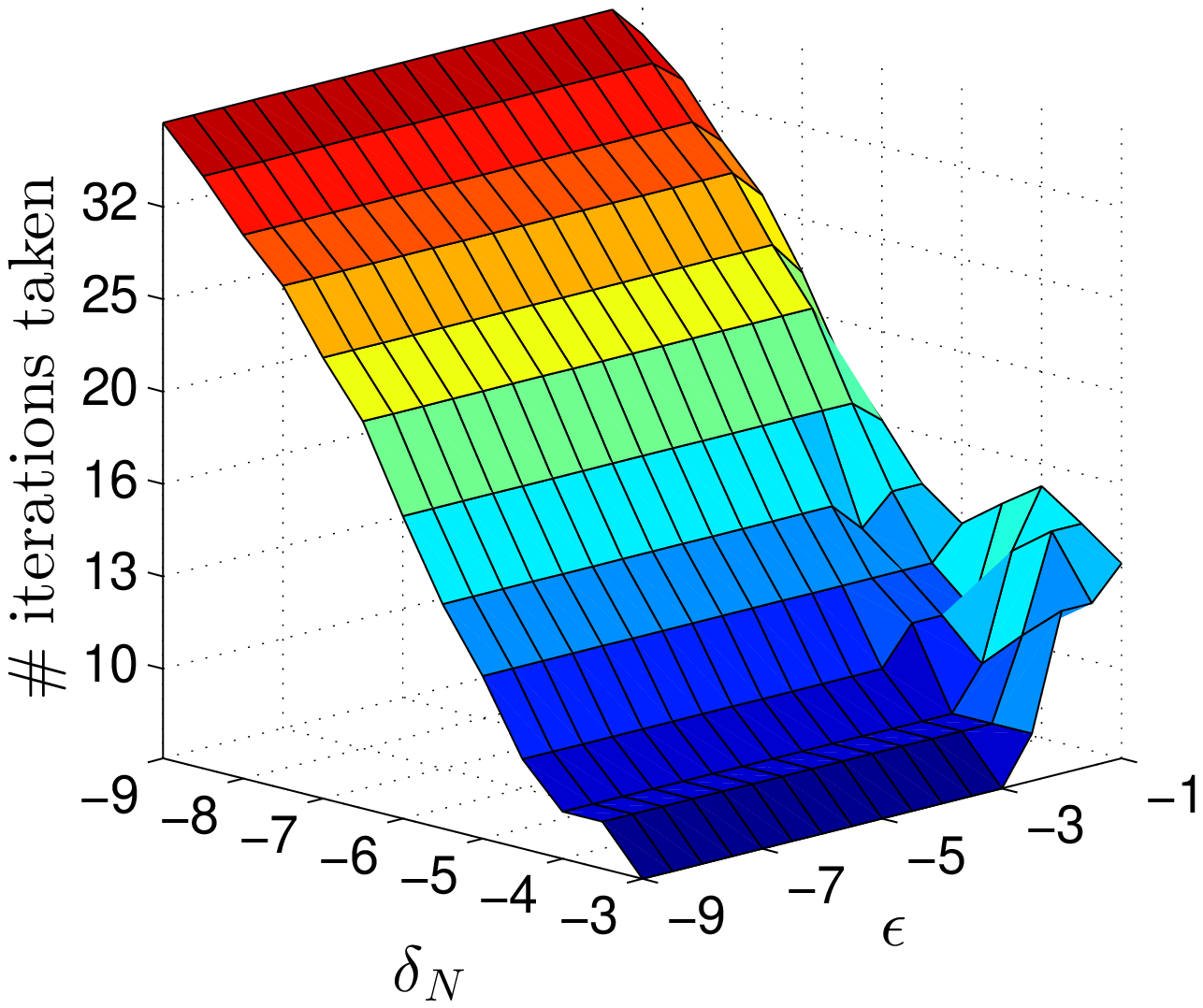}\label{fig:epsDelta2a} }
    \subfloat[]{\includegraphics[width=0.4\textwidth]{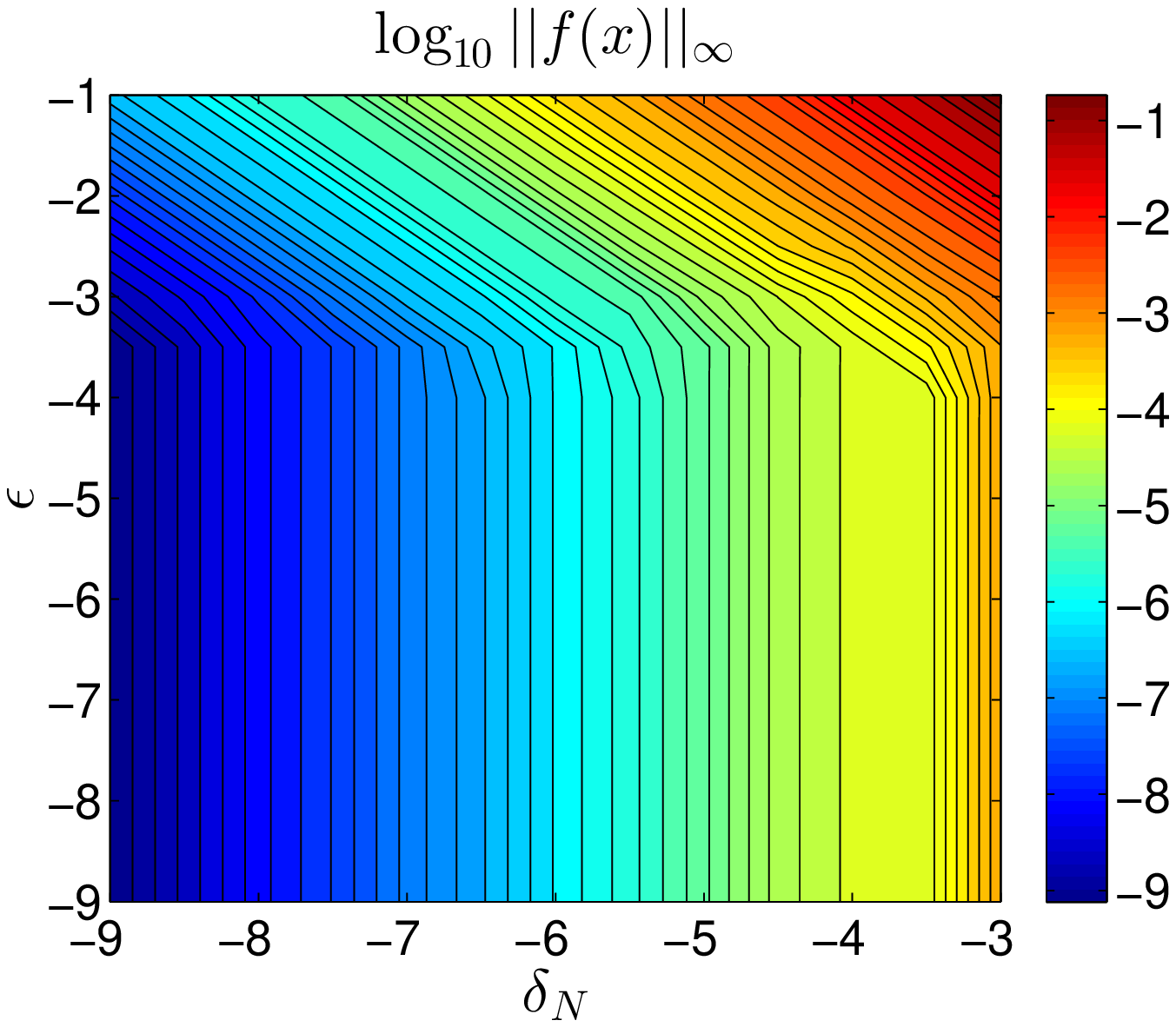} \label{fig:epsDelta2b}}
    \caption{A  parameter sweep for $\epsilon$ and $\delta_N$  against (a) the number of iterations (b) residual norm of the nonlinear equation, $\|f(x)\|_\infty$, for network EXnet. }%
    \label{fig:epsDelta2}%
\end{figure*}

\subsection{Stopping criteria for the null space algorithm}\label{subsec:stopcriteria}

The nulls pace method of Algorithm~\ref{alg:null_full} requires the satisfaction of the set tolerance $\delta_N$ to stop, provided the nonlinear equation residual inequality $\|f(q^k,h^k)\|\leq \delta_N$ can be achieved under the maximum number of iterations allowed. Although some have used the convergence of the flow conservation residual, $\|A_{12}^Tq-d\|_\infty$, as a stopping criteria, recent literature~\cite{elhay2011dealing,kovalDiscussElhay2013} has shown for the GGA method that the flow conservation equation often converges to within machine precision many iterations before the energy residual becomes sufficiently small. For Algorithm~\ref{alg:null_full}, flow conservation is actually satisfied by all Newton iterates; see~\eqref{eq:nullspace_sum}. Therefore, the convergence of the flow continuity equation should not be used as a stopping criteria. It is necessary to compute the head to determine convergence using either the residual of the entire nonlinear equation~\eqref{eq:hydro_nl_eqn},  or convergence of nodal head differences at each iteration, as also proposed for the GGA method in~\cite{elhay2011dealing,kovalDiscussElhay2013}. 

Computing the pressure heads at each iteration by solving~\eqref{eq:nullspace_lineq_a} and the associated norm of the residual, together, add significant computational cost. However, unlike for the GGA method,  the flow iterations are independent of the head values in the null space formulation -- see~\eqref{eq:nullspace_lineq} and~\eqref{eq:nullspace_lineq_a}, or Algorithm~\ref{alg:null_full}. This brings the possibility that we can delay head computations until near convergence, where pressure heads can be computed to check convergence of  the residual.
  \begin{figure}[]%
    \centering
    \subfloat{{\includegraphics[width=0.5\textwidth]{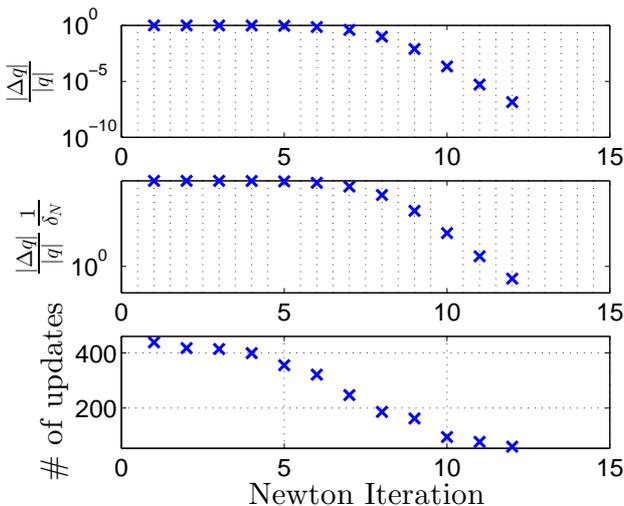} }}
\caption{Convergence of various variables with Newton iteration $k$. Top: Middle:  Bottom: the number of flows from the set $\mathcal{E}_2$ that remain in the update set $\mathcal{U}^k(\epsilon,\delta_N),$ where $\delta_N=1e-6$, $\epsilon=1e-3.$}
    \label{fig:convdq_Richnet}%
\end{figure} 
Traditionally, in open source software like EPANET, a pragmatic convergence criterion is applied based on the sum of all flow changes as a proportion of the total flow rates in all links~\cite[pp.\ 153]{rossman2000epanet}. The top plot in Figure~\ref{fig:convdq_Richnet} shows the ratio $\frac{\|\Delta q\|_1}{\|q\|_1}$ for the network Richnet. By default, EPANET uses $0.001$ for this number as a convergence criteria. We also plot the ratio of this number to the Newton tolerance set in the middle plot of Figure~\ref{fig:convdq_Richnet}. From this and similar plots at different values for $\delta_N$ for all the models, this ratio becomes less than 1 near convergence. In agreement with our discussions of Figure~\ref{fig:numUpdates}, the bottom plot in Figure~\ref{fig:convdq_Richnet} also shows that size of the update set diminishes toward zero, i.e.\ the update set $\mathcal{U}^k$ is significantly smaller than $\mathcal{E}_2$ near convergence. Therefore, we can reduce the overhead in computing the pressure heads and  error norm computations by computing them only when the fraction of non-converged flows is smaller than the set $\mathcal{E}_2$. Therefore, we can reduce the overhead in computing the pressure heads and  error norm computations by starting such computations  only when either the fraction of non-converged flows is significantly smaller than the set $\mathcal{E}_2$ or when the ratio $\frac{\|\Delta q\|_1}{\|q\|_1}$ is less than $\delta_N.$ In the next section, we use the heuristic condition $|\mathcal{U}^k|<\lceil a\times |\mathcal{E}_2|\rceil, \; a=0.5 $ OR $\frac{\|\Delta q\|_1}{\|q\|_1 } \leq \delta_N$ to reduce computation time significantly for all example network models. 
\section{Simulation Results}\label{sec:caseStudy}
We use seven networks, some of which are proprietary operational water network models, to analyse the null space  method we have proposed. The networks range in size from 444 pipes to 14,831 pipes and have varying levels of `loopedness' as measured by the ratio of loops to total number of pipes ($\frac{n_l}{n_p}$) or the average degree of the graph, i.e.\ the average number of pipes incident at each node. The basic properties of the case study networks and their relevant topological characteristics are given in Table~\ref{table:CSnetworks}, ordered by increasing size.  The sparsity of the incidence matrices are  around 3 for all these sparse network graphs. We note that, generally, the densest of water networks are still sparse in the mathematical sense; for example, compare with social and economic networks that can have orders of magnitude larger degrees~\cite{jackson2008social}. The proprietary operational models BWFLnet and BWKnet~\cite{abrahamnull2014} are from a typical network in a built up (urban) area in England, UK. They are parts of a distribution system used in experimental studies of dynamic sectorization and optimal pressure control of water supply systems by the InfraSense Labs in partnership with a UK water utility~\cite{Wright2014}. 

The networks Richnet (a medium-sized real network from Yokshire Water, UK~\cite{van2004operational}), WCnet (Wolf-Cordera, part of a real network in Colorado Springs, USA~\cite{lippai2005colorado}), EXnet (an artificial network for design and rehabilitation optimization that has a large number of triangular and trapezoidal loops~\cite{Farmani2004}) and NYnet (a approximately all-pipe model of a real network from~\cite{ostfeld2008battle}) are the ones analysed 
in~\cite{elhay2014reformulated}. The relatively smaller size artificial network  C-town~\cite{ostfeld2011battle}	, called  CTnet here, is also used. 
 
%

\begin{table}[t!]
{\small
\caption{ Size and graph characteristics of the different case study networks; $\texttt{incMat}$ denotes the incidence matrix for the vertices of a network's graph and \texttt{deg} represents the graph's average degree (i.e.\ $\texttt{deg}=2n_p/n_n$) and . }
\label{table:CSnetworks}
\begin{center}
\tabcolsep=0.13mm
\begin{tabular}{l|l|l|l|l|l|l|c|c}
\cline{1-9}
Network & $n_p$	& $n_n$ & $n_l$& $\frac{n_l}{n_p}\%$	&$n_0$	& {$\texttt{deg}$} & $\frac{
nnz(\texttt{incMat})}{n_n}$&{Headloss }\\
\cline{1-9}
CTnet   & 444	& 388	&  48	&10.8	&8	&  2.24 & 3.25   		&  HW\\
\cline{1-9}
Richnet & 934	& 848	& 86	&9.2	&8	& 2.20 	&3.17    	&  DW\\
\cline{1-9}
WCnet   & 1976	& 1770	& 206	&10.4 	&4	& 2.23	&3.22   	&  DW \\
\cline{1-9}
BWFLnet	& 2369	& 2303	& 66	&2.8 	&2	& 2.05 	&3.05	&  HW\\
\cline{1-9}
EXnet	&2465	& 1890	& 575	&23.3	&3	& 2.61	&3.55  		& DW \\
\cline{1-9}
BWKWnet	& 4648	& 4577	& 71	&1.5 	&1	& 2.03	&3.03  	& HW \\
\cline{1-9}
NYnet	&14830	& 12523	& 2307	&15.6 	&7	& 2.37&3.29 & DW \\
\cline{1-9}
\end{tabular}
\end{center}
}
\end{table}

All computations were performed within MATLAB R2013b-64 bit for Windows 7.0 installed on  a 2.4 GHz Intel\textsuperscript{\textregistered} Xeon(R) CPU E5-2665 0 with 16 Cores.  To make the CPU time profiling most accurate, the number of active CPUs used by Matlab was set to one before starting profiling. This prevents spurious  results from  the use of multiple cores by some of the solvers used. For example, the approximate minimum ordering (AMD) and its variants (minimum fill, column minimum degree ordering, etc.\ ) and graphs-based permutations used in the sparse Cholesky, LU and QR factorizations and solves, within Matlab and SuiteSparse, take advantage of parallelizing work over multiple cores; these should be disabled to make a fairer comparison of the proposed algorithms. Moreover, a large number of simulations (1000) were used to analyse each case study because small variations in task scheduling by the processor could result in variations not caused by computational complexity only. 
The numerical tests were performed by randomly varying the demands from the typical diurnal demand profile. As in~\cite{elhay2014reformulated} and other referenced literature, all analysis presented here do not consider control devices like pumps and check valves.
The method for computing the  Darcy-Weisbach resistance coefficients  was written in C++ and implemented as a MATLAB MEX-function, which has an execution time similar to a C++ implementation. 

To reuse the symbolic factors of the Cholesky factorization in~\eqref{eq:nullspace_lineq} for the simulations in Table~\ref{table:Singlesymfact}, the SimplicialLLT sparse Cholesky module of Eigen 3.2.1~\cite{jacob2012eigen} was used in a proprietary C++ implementation of the null space method of Algorithm~\ref{alg:null_full}. 
This implementation decouples the linear solve into  \texttt{analyze()}, \texttt{factorize()} and \texttt{solve()} steps. The analysis step applies the AMD preordering~\cite{amestoy2004algorithm} followed by a symbolic factorization on the sparsity of $Z^TZ$, which is the same constant structure used for all iterations.  The \texttt{factorize()} and \texttt{solve()} functions perform a numeric decomposition of matrix $Z^TF^kZ$ and the solution by substitution, respectively, at each Newton iteration.  
 For all presented tests,  the computational times can vary with hardware and operating systems. The trends in the results, nonetheless, should be valid generally. Future work includes the implementation of these methods in C++.

\begin{table*}[ht!]
\caption{ Mean CPU times for the Schur and null space methods applied to networks of different size and connectivity; the accuracy and partial update set parameter were set to $\delta_N=1e^{-6}$ and $\epsilon=1e^{-3}$, respectively. }
\label{table:cputimes}
\begin{center}
\begin{tabular}{|l|l|l|l|l|l|l|l|}
\cline{1-8}
\multicolumn{1}{|c|}{} &\multicolumn{4}{c|}{CPU times (ms)}& \multicolumn{3}{c|}{$\frac{t(\text{Schur})}{t(\text{NSM})}$ }\\
\cline{2-8}
Network & Schur&NSM1 & NSM2 & NSM3 &NSM1 & NSM2 & NSM3 \\
\cline{1-8}
CTnet   & 9.34	& 6.95	&  6.30	& 5.27	& 1.34  & 1.48  & 1.77  \\
\cline{1-8}
Richnet  &15.04   &10.81 &  8.29    &6.93	&  1.39&1.82   &2.17  \\
\cline{1-8}
WCnet    &  26.27  &20.87  &18.50   &16.22& 1.26 & 1.42 & 1.61  \\
\cline{1-8}
BWFLnet    &    13.65&    7.65&    6.45    &5.53 & 1.78 & 2.11 & 2.46  \\
\cline{1-8}
EXnet	 &     87.90&   82.10&   58.42&   55.31& 1.07 & 1.50 & 1.59  \\
\cline{1-8}
BWKWnet	 &23.83&  11.67&    8.07&   6.28& 2.04 & 2.95 & 3.79  \\
\cline{1-8}
NYnet	 &   549.81&  512.12&  370.80& 347.71& 1.07 & 1.48 & 1.58  \\
\cline{1-8}
\end{tabular}
\end{center}
\end{table*}
Table~\ref{table:cputimes} presents a comparison of  the null space algorithm as described in Algorithm~\ref{alg:null_full}, called NSM1 here, with its modified versions with our proposed partial update scheme only (NSM2), and one with both a partial update scheme and  head loss computations that start near convergence using the proposed heuristics in Section~\ref{subsec:stopcriteria}(NSM3). 

The results of Table~\ref{table:cputimes} demonstrate the trends observed in Figure~\ref{fig:schurvsNullCPU_1}. The null space algorithms reduce average CPU time for all the given networks, the highest being by almost a factor of 4 for BWKnet. As expected from Algorithm~\ref{alg:null_full}, a null space method have the biggest impact in reducing computational cost when the network is not highly looped, i.e.\ $n_c<<n_p$. This is apparent from the results --  the least dense networks, BWFLnet and BWKnet, have the highest reduction in CPU time.  For the most meshed networks, EXnet and NYnet, the null space algorithm NSM1 have the smallest relative reduction in CPU time. 
From Table~\ref{table:Zproperties}, we note that the Newton  linear systems of the null space  method become bigger and less sparse the more meshed a network is. These result in less savings from the linear solve stage of the null space algorithm compared to for the sparser networks. 
Moreover, the networks with higher average degrees also have fundamental null bases with a  higher number of links involved in the loop equations;  the bigger size of the the update set $E_2$ becomes, the less are the savings gained when applying partial headloss computations. 
The trends for the null space algorithms NSM2 and NSM3  demonstrate  the additional relative savings made using our novel partial updates and the new heuristic to delay computing pressure head  values until the algorithm is near convergence, respectively. For all network models, significant additional savings are made by the novel approaches of NSM2 and NSM3 compared to the  null space method of Algorithm~\ref{alg:null_full} (NSM1). 

\section{Conclusion}\label{sec:conclusion}

In order  to facilitate the reliable and  efficient near real-time management of water distribution systems, we have analyzed the use of a null space inexact Newton method for demand driven hydraulic simulations of large scale water distribution networks. 
The saddle point structure  of the Jacobian in the Newton linear systems has been exploited to describe and propose novel sparse null space approaches, which  solve the  nonlinear hydraulic  equations  with less computational resources and more robustly than the equivalent Schur (or GGA) approach. Having described various methods for formulating and solving hydraulic equations, we have proposed techniques for increasing computational efficiency of a null space algorithm. We have presented a study of algorithms for generating null bases with respect to the sparsity,  condition number and  the fraction of total links involved in the null space loop equations. Using simulation results from an operational network model, we have shown the ubiquity of zero flows, and the inherent bad conditioning of the resulting linear systems for models with a range of scales in pipe diameters; a Jacobian regularization scheme from~\cite{elhay2011dealing} has also been adopted to improve the condition number of the linear systems.  Since the nulls pace projected linear systems have a Jacobian with a constant sparsity structure, symbolic factorization of the Cholesky solvers can be reused. 
We have demonstrated using our case study networks that, for the more meshed networks where the linear solve times take a large fraction of the Newton method CPU time, significant computational savings  can be made by reusing the symbolic Cholesky factors.

The repeated headloss computations for both Hazen-Williams and Darcy-Weisbach models take a significant fraction of total flops used by the Newton iterations. We show that only a fraction of the network graph edges are projected into the null space when appropriate  fundamental null space basis are used. Therefore, headlosses need only be computed for these links, reducing computational cost.  Moreover, many of the flow values for links involved in the loop equations converge well before the end of the Newton iterations. 
A partial update set, with size that diminishes with Newton iterations, is proposed as an index set so that headlosses are updated only for loop flows that have not yet converged;  this  has been shown to further reduce computational cost.  The parametrization of the update sets is studied to propose appropriate values.
A proof is given to guarantee the convergence of the inexact Newton method under partial updates.

Since the  flow iterates generated by the null space Newton method do not depend on pressure head values, the linear systems solved to compute pressure heads can be delayed until near convergence.
Based on the relative size of the partial update sets and relative norm of flow changes, we have proposed a  heuristic to avoid computing pressure head  values at each Newton iteration. This has been shown to reduce computational cost further. We have  used case studies with both synthetic network models from literature and large scale models of operational water distribution networks, of various sizes and meshedness,  to demonstrate the effectiveness  of our novel null space approaches. Results show that, for the sparsest of the example operational networks, CPU time for our efficient null space approach is reduced by nearly a factor of  4  compared to an efficient Schur method.

\section*{References}
\bibliography{paper2_biblio}

\end{document}